\documentclass[12pt]{article}
\usepackage{amssymb,amsthm,amsmath,latexsym}
\usepackage{url}
\usepackage{color}
\usepackage{comment}
\usepackage{lscape}
\newtheorem{thm}{Theorem}
\newtheorem{prop}[thm]{Proposition}
\newtheorem{lem}[thm]{Lemma}
\newtheorem{cor}[thm]{Corollary}

\theoremstyle{remark}
\newtheorem{rem}[thm]{Remark}
\newcommand{\FF}{\mathbb{F}}

\newcommand{\ww}{\omega}
\newcommand{\vv}{\overline{\omega}}
\newcommand{\cC}{\mathcal{C}}

\DeclareMathOperator{\wt}{wt}

\begin{document}
\title{A method for constructing quaternary Hermitian self-dual codes
and an application to quantum codes}

\author{
Masaaki Harada\thanks{
Research Center for Pure and Applied Mathematics,
Graduate School of Information Sciences,
Tohoku University, Sendai 980--8579, Japan.
email: \texttt{mharada@tohoku.ac.jp}.}
}

\maketitle

\begin{abstract}
We introduce quaternary modified four $\mu$-circulant codes
as a modification of four circulant codes.
We give basic properties of quaternary modified four $\mu$-circulant 
Hermitian self-dual codes.
We also construct quaternary modified  four $\mu$-circulant Hermitian self-dual codes
having large minimum weights.
Two quaternary  Hermitian self-dual $[56,28,16]$ codes are constructed
for the first time.
These codes improve the previously known lower bound on the largest minimum weight among
all quaternary (linear) $[56,28]$ codes.
In addition, these codes imply the existence of a quantum $[[56,0,16]]$ code. 
\end{abstract}

%%%%%%%%%%%%%%%%%%%%%
\section{Introduction}
Self-dual codes are one of the most interesting classes of  (linear)  codes.
This interest is justified by many combinatorial objects
and algebraic objects related to self-dual codes
(see e.g.,~\cite{SPLAG}, 
\cite{Hu05} and~\cite{RS-Handbook}).

Let $\FF_{q^2}$ denote the finite field of order ${q^2}$, where $q$ is a prime or a prime power.
A code $C$ over $\FF_{q^2}$ of length $n$ is said to be {Hermitian self-dual} if
$C=C^{\perp_H}$, where
the Hermitian dual code $C^{\perp_H}$ of $C$ is defined as
$C^{\perp_H}=\{x \in \FF_{q^2}^{n} \mid \langle x,y\rangle_H=0 \text{ for all } y\in C\}$
under the Hermitian inner product $\langle x,y\rangle_H$.
%Codes over $\FF_2,\FF_3$ and $\FF_4$ are called \emph{binary}, \emph{ternary} and
%\emph{quaternary}, respectively.
By the Gleason--Pierce theorem, there are nontrivial divisible Hermitian self-dual
codes over $\FF_{q^2}$ for $q=2$ only.
This is one of the reasons why
much work has been done concerning Hermitian self-dual codes over $\FF_4$
(see e.g.,~\cite{AH},
\cite{CP},
\cite{CPS},
\cite{GG},
\cite{G00},
\cite{GHM},
\cite{HLMT},
\cite{HM11},
\cite{Hu90},
\cite{Hu91},
\cite{KL},
\cite{K01},
\cite{LP},
\cite{MMS},
\cite{MOSW} and
\cite{Roberts}).
In this paper, we study Hermitian self-dual codes over $\FF_4$.

It is a fundamental and challenging problem in self-dual codes to classify self-dual codes and
determine the largest minimum weight among all self-dual codes for a fixed length.
A code over $\FF_4$ is called {quaternary}.
All quaternary Hermitian self-dual codes were classified 
in~\cite{CPS},~\cite{HLMT},~\cite{HM11} and~\cite{MOSW}
for lengths $n \le 20$.
% The minimum weight $d$ of a quaternary Hermitian self-dual
% code of length $n$ is bounded by
% $d \leq 2 \lfloor n/6 \rfloor +2$~\cite{MOSW}.
% A quaternary Hermitian self-dual code of length $n$ and minimum weight
% $2 \lfloor n/6 \rfloor +2$ is called \emph{extremal}.
Also, the largest minimum weight $d(n)$ among all Hermitian self-dual codes 
is determined for lengths $n \le 30$
(see~\cite[Table~5]{GG} for the current information on $d(n)$).

For small fields $\FF$, 
many four circulant  (negacirculant) self-dual codes over $\FF$
having large minimum weights are known
(see e.g.,~\cite{GL},
\cite{H20},
\cite{H23},
\cite{HHKK} and the references given therein).
In this paper, 
by modifying four circulant self-dual codes,
we give a method for constructing quaternary Hermitian self-dual  codes based on $\mu$-circulant
matrices, which are called modified four $\mu$-circulant codes.
Some basic properties of modified four $\mu$-circulant quaternary Hermitian self-dual  codes
are given.
We also  give numerical results of 
quaternary modified  four $\mu$-circulant Hermitian self-dual codes
together with an application to quantum codes.

This paper is organized as follows.
In Section~\ref{Sec:2}, we give some definitions, notations and basic results used in this paper.
In Section~\ref{Sec:Def}, we  define quaternary modified four $\mu$-circulant codes
as a certain modification of four circulant codes.
We also give basic properties of quaternary modified four $\mu$-circulant 
Hermitian self-dual codes.
In particular, we give a condition for quaternary modified four $\mu$-circulant codes
to be Hermitian self-dual.
In addition, we observe equivalences of quaternary modified  four $\mu$-circulant Hermitian self-dual codes.
% (Lemmas~\ref{lem:B:SD},~\ref{lem:B:equiv} and~\ref{lem:B:1}).
In Section~\ref{Sec:N}, we present numerical results of 
quaternary modified  four $\mu$-circulant Hermitian self-dual codes.
By computer search based on  basic properties  presented in Section~\ref{Sec:Def},
we give a classification of 
quaternary modified four $\mu$-circulant Hermitian self-dual codes 
having the currently known largest minimum weights
for lengths $24,28,32$ and $36$ (Proposition~\ref{prop:classification}).
For larger lengths,
we also construct quaternary modified four $\mu$-circulant Hermitian self-dual codes 
having large minimum weights.
We emphasize that quaternary  Hermitian self-dual $[56,28,16]$ codes are constructed
for the first time  (Proposition~\ref{prop:56}).
These codes  $C_{56,1}$ and $C_{56,\ww}$ improve the previously known lower bounds on the largest minimum weight among
all quaternary (linear) $[56,28]$ codes (Corollary~\ref{cor:56}).
In Section~\ref{Sec:Q}, we give an application of $C_{56,1}$ and $C_{56,\ww}$
to quantum codes.  
More precisely, $C_{56,1}$ and $C_{56,\ww}$ imply
the existence of a quantum $[[56,0,16]]$ code.  
%This quantum code improves the previously known lower bound on 
%the largest minimum weight among quantum $[[56,0,d]]$ codes (Proposition~\ref{prop:Q}).

%%%%%%%%%%%%%%%%%%%%%%%%%%%%%%%%%%%%%%
\section{Preliminaries}\label{Sec:2}

In this section, we give some definitions, notations
and basic results used in this paper.

%%%%%%%%%%%%%%%%%%%%%
\subsection{Quaternary codes}

We denote the finite field of order $4$
by $\FF_4=\{ 0,1,\ww , \vv  \}$, where $\vv= \omega^2 = \omega +1$.
A \emph{quaternary} linear $[n,k]$ \emph{code} $C$
is a $k$-dimensional vector subspace of $\FF_4^n$.
All codes in this paper are quaternary and linear unless otherwise noted, 
so we omit linear and we often omit quaternary. 
% All codes in this paper are quaternary unless otherwise noted.
The parameter $n$ is called the \emph{length} of $C$.
A \emph{generator matrix} of $C$ is a $k \times n$
matrix such that the rows of the matrix generate $C$.
The \emph{weight} $\wt(x)$ of a vector $x \in \FF_4^n$ is
the number of non-zero components of $x$.
%The \emph{support} $\supp(x)$ of a vector $x=(x_1,x_2,\ldots,x_n) \in \FF_4^n$ is
%a subset $\{i \mid x_i \ne 0\}$ of $\{1,2,\ldots,n\}$.
%The \emph{weight} $\wt(x)$ of a vector $x$ is defined as $|\supp(x)|$.
The \emph{weight enumerator} of $C$ is given by $\sum_{c \in C} y^{\wt(c)}$.
A vector of  $C$ is called a \emph{codeword} of $C$.
The minimum non-zero weight of all codewords in $C$ is called
the \emph{minimum weight} of $C$.  
A  \emph{quaternary  $[n,k,d]$ code}
is a  quaternary $[n,k]$ code with minimum weight $d$.

%%%%%%%%%%%%%%%%%%%%%
\subsection{Quaternary Hermitian self-dual codes}

The \emph{Hermitian dual} code $C^{\perp_H}$ of a quaternary
code $C$ of length $n$ is defined as
\[
C^{\perp_H}=
\{x \in \FF_{4}^n \mid \langle x,y\rangle_H = 0 \text{ for all } y \in C \},
\]
under the following Hermitian inner product
\[
\langle x,y\rangle_H= \sum_{i=1}^{n} x_i y_i^2
\]
for $x=(x_1,x_2,\ldots,x_n)$, $y=(y_1,y_2,\ldots,y_n)\in \FF_{4}^n$.
A quaternary code $C$ is said to be \emph{Hermitian self-dual} if
$C=C^{\perp_H}$.
All codewords of a quaternary Hermitian self-dual code have even weights~\cite[Theorem~1]{MOSW}.
%A quaternary Hermitian self-dual code of length $n$ exists
%if and only if $n>0$ with $n \equiv 0 \pmod 2$. 

All matrices in this paper are matrices over $\FF_4$, so we write simply matrices.
Throughout this paper,
let $I_n$ denote the identity matrix of order $n$,
and let
$A^T$ denote the transpose of a matrix $A$.
Moreover, let $\overline{A}$ denote the matrix $(a_{ij}^2)$ for a matrix $A=(a_{ij})$.
The following lemma is a criterion for Hermitian self-duality.
\begin{lem}[{\cite[Theorems~1 and 4]{MOSW}}]
\label{lem:SD}
Let $C$ be a quaternary $[2n,n]$ code with generator matrix
$
\left(
\begin{array}{cc}
I_{n} & M 
\end{array}
\right)$.
If $M \overline{M}^T=I_n$, then $C$ is Hermitian self-dual.
\end{lem}

It was shown in~\cite{MOSW} that the minimum weight $d$ of a quaternary Hermitian self-dual
code of length $n$ is bounded by:
\begin{equation}\label{eq:bound}
d \leq 2 \left\lfloor \frac{n}{6}\right\rfloor +2.
\end{equation}
A  quaternary Hermitian self-dual code of length $n$ and minimum weight
$2 \lfloor n/6 \rfloor +2$ is called \emph{extremal}.

Two quaternary  Hermitian self-dual  codes $C$ and $C'$ are \emph{equivalent} if there is a
monomial matrix $P$ over $\mathbb{F}_4$ with $C' = C \cdot P$,
where $C \cdot P = \{ x P\mid  x \in C\}$ (see~\cite{MOSW}).
Throughout this paper, two equivalent quaternary  Hermitian self-dual codes $C$ and $C'$ are denoted by $C \cong C'$.
All quaternary Hermitian self-dual codes were classified 
in~\cite{CPS},
\cite{HLMT},
\cite{HM11} and~\cite{MOSW}
for lengths up to $20$.
All extremal  quaternary Hermitian self-dual codes of length $22$ were also classified 
in~\cite{HM11}.

%%%%%%%%%%%%%%%%%%%%%
\section{Definition and basic properties of modified four $\mu$-circulant codes}
%Four-$\ww$-circulant codes and four $\vv$-circulant codes}
\label{Sec:Def}

In this section, we  define quaternary modified four $\mu$-circulant codes
and we give their basic properties.

% An $n \times n$ \emph{circulant} matrix has the form
% \[
% \left( \begin{array}{cccccc}
% r_0&r_1&r_2& \cdots        &r_{n-2} &r_{n-1} \\
% r_{n-1}&r_0&r_1& \cdots &r_{n-3} &r_{n-2} \\
% r_{n-2}&r_{n-1}&r_0& \cdots &r_{n-4} &r_{n-3} \\
% \vdots &\vdots & \vdots &&\vdots& \vdots\\
% r_1&r_2&r_3& \cdots&r_{n-1}&r_0
% \end{array}
% \right),
% \]
% and an $n \times n$ \emph{$\mu$-circulant} matrix has the form

An $n \times n$  matrix of the following form
\[
\left( \begin{array}{cccccc}
r_0&r_1&r_2& \cdots &r_{n-2} &r_{n-1}\\
\mu r_{n-1}&r_0&r_1& \cdots &r_{n-3}&r_{n-2} \\
\mu r_{n-2}&\mu r_{n-1}&r_0& \cdots &r_{n-4}&r_{n-3} \\
\vdots &\vdots & \vdots &&\vdots& \vdots\\
\vdots &\vdots & \vdots &&\vdots& \vdots\\
\mu r_1&\mu r_2&\mu r_3& \cdots&\mu r_{n-1}&r_0
\end{array}
\right)
\]
is called \emph{$\mu$-circulant}, where $\mu \in \{1,\ww,\vv\}$.
In particular, 
if $\mu =1$, then this is well-known as a \emph{circulant} matrix.
It is trivial that a $\mu$-circulant matrix with first row
$(r_0,r_1,\ldots, r_{n-1})$ is written as
$\sum_{i=0}^{n-1} r_i E_n(\mu)^i$, where
\[
E_n(\mu)=
\left( \begin{array}{cccccc}
0 &  &  &  \\
\vdots &  & I_{n-1} &  \\
0 & &  &  \\
\mu & 0 & \cdots & 0 \\
\end{array}
\right).
\]
%(see~\cite[Section~2.3]{Roberts}).

%Although the following lemma is somewhat trivial, we give a proof for the sake of completeness.

\begin{lem}\label{lem:B:ABeqBA}
Suppose that $\mu \in \{1,\ww,\vv\}$.
\begin{enumerate}
\item
If $A$ and $B$ are $n \times n$ $\mu$-circulant matrices, then $AB=BA$.
\item
If $A$ is an $n \times n$ $\mu$-circulant matrix with first row
$(r_0,r_1,\ldots, r_{n-1})$, then $\overline{A}^T$ is a 
$\mu$-circulant matrix with first row
$(r_0^2, (\mu r_{n-1})^2,\ldots, (\mu r_1)^2)$.
\end{enumerate}
\end{lem}
\begin{proof}
The assertion (i) follows from the fact that a  $\mu$-circulant matrix with first row
$(r_0, r_1,\ldots,r_{n-1})$ is written as 
$\sum_{i=0}^{n-1} r_i E_n(\mu)^i$.
%(see also~\cite[Section~2.3]{Roberts}).
The assertion (ii) follows from the fact that 
$\overline{A}^T$ is written as 
$r_0^2I_n + \sum_{i=1}^{n-1} (\mu r_{n-i})^2 E_n(\mu)^i$.
\end{proof}

% Let $A$ and $B$ be $n \times n$ circulant matrices.
% A quaternary $[4n,2n]$ code with generator matrix of the following form:
% \[
% \left(
% \begin{array}{ccc@{}c}
% \quad & {\Large I_{2n}} & \quad &
% \begin{array}{cc}
% A & B \\
% B^T & A^T
% \end{array}
% \end{array}
% \right)
% \]
% is called a \emph{four circulant} code.
% For small fields $\FF$, 
% many four circulant self-dual codes over $\FF$
% having large minimum weights are known
% (see e.g.,~\cite{G00},~\cite{GH},~\cite{GHM} and~\cite{HHKK}).

By modifying four circulant self-dual codes 
(see e.g.,~\cite{HHKK} for the definition), 
%of four circulant self-dual codes), 
we introduce the following method for constructing quaternary Hermitian self-dual
codes.
Suppose that $\mu \in \{1,\ww,\vv\}$.
Let $A$ and $B$ be $n \times n$ $\mu$-circulant matrices.
We say that a quaternary $[4n,2n]$ code with generator matrix of the following form
\begin{equation}\label{eq:const}
\left(
\begin{array}{ccc@{}c}
\quad & {\Large I_{2n}} & \quad &
\begin{array}{cc}
A & B \\
\overline{B}^T & \overline{A}^T
\end{array}
\end{array}
\right)
\end{equation}
is \emph{modified four $\mu$-circulant}.
A modified four $1$-circulant code is also called \emph{modified four circulant}.
% This code is simply called \emph{m-four $\mu$-circulant}.
We denote the code with generator matrix~\eqref{eq:const} by $C_{\mu}(A,B)$.

\begin{rem}
As a different modification of four circulant codes,
codes with generator matrices of the following form
\[
\left(
\begin{array}{ccc@{}c}
\quad & {\Large I_{2n}} & \quad &
\begin{array}{cc}
A^TCJ & \overline{B} \\
B^TCJ &  \overline{A}
\end{array}
\end{array}
\right)
\]
are given in~\cite{Roberts},
where $A,B$ and $C$ are circulant matrices and $J$ is the exchange matrix.
\end{rem}

Now we give some basic properties of modified  four $\mu$-circulant
 Hermitian self-dual codes.
Although the following lemmas are somewhat trivial, we give proofs for the sake of completeness.

\begin{lem}\label{lem:B:SD}
Suppose that $\mu \in \{1,\ww,\vv\}$.
A quaternary modified four $\mu$-circulant code $C_{\mu}(A,B)$
is Hermitian self-dual
if $A \overline{A}^T +B \overline{B}^T =I_n$.
\end{lem}
\begin{proof}
By Lemma~\ref{lem:B:ABeqBA} (i), $AB+BA=O_n$, where
$O_n$ denotes the $n \times n$ zero matrix.
By Lemma~\ref{lem:B:ABeqBA} (ii), $\overline{A}^T$ and
$\overline{B}^T$ are $\mu$-circulant.
Again by Lemma~\ref{lem:B:ABeqBA} (i), 
$\overline{A}^T \overline{B}^T=\overline{B}^T\overline{A}^T$
and $A\overline{A}^T =\overline{A}^TA$.
Thus, we have 
\[
\overline{A}^T\overline{B}^T+\overline{B}^T\overline{A}^T =O_n
\text{ and } 
\overline{A}^TA + \overline{B}^TB=I_n.
\]
%In addition, 
%$(\overline{\overline{A}^T\overline{B}^T+\overline{B}^T\overline{A}^T})^T
%=BA+AB=O_n$.
%Thus, it holds that
%$\overline{A}^T\overline{B}^T+\overline{B}^T\overline{A}^T =O_n$.
Let  $M(A,B)$ denote the $2n \times 2n$ matrix
$
\left(
\begin{array}{cc}
A & B \\
\overline{B}^T & \overline{A}^T
\end{array}
\right)$.
Then we have
\[
M(A,B) \overline{M(A,B)}^T=
\left(
\begin{array}{cc}
A \overline{A}^T +B \overline{B}^T & AB+BA \\
\overline{A}^T\overline{B}^T+\overline{B}^T\overline{A}^T 
& \overline{A}^TA + \overline{B}^TB
\end{array}
\right)
=I_{2n}.
\]
The result follows from Lemma~\ref{lem:SD}.
\end{proof}

\begin{lem}\label{lem:B:equiv}
Suppose that $C_{\mu}(A,B)$ is a quaternary modified four $\mu$-circulant Hermitian self-dual 
code, where $\mu \in \{1,\ww,\vv\}$.
Then the following statements hold.
\begin{enumerate}
\item
$C_{\mu}(A,B) \cong C_{\mu}(\ww A,\ww B) \cong C_{\mu}(\vv A,\vv B)$.
\item
$C_{\mu}(A,B) \cong C_{\mu}(B,A)$.
\item
$C_{\mu}(A,B) \cong C_{\mu}(\overline{A}^T, \overline{B}^T)$.
\item
$C_{\mu}(A,B) \cong C_{\mu}(A,\overline{B}^T)$.
\end{enumerate}
\end{lem}
\begin{proof}
The assertions (i), (ii) and (iii) are trivial.
The Hermitian dual code
$C_{\mu}(A,B)^{\perp_H}$ of $C_{\mu}(A,B)$
has the following generator matrix
\[
\left(
\begin{array}{cccc}
\begin{array}{cc}
\overline{A}^T & B \\
\overline{B}^T & A
\end{array}
\quad & {\Large I_{2n}} & \
\end{array}
\right).
\]
Since $C_{\mu}(A,B)=C_{\mu}(A,B)^{\perp_H}$,
the above matrix is also a generator matrix of $C_{\mu}(A,B)$.
It follows from (iii) that $C_{\mu}(A,B) \cong C_{\mu}(A,\overline{B}^T)$.
\end{proof}

\begin{lem}\label{lem:B:1}
Let $C$ be a quaternary modified four $\mu$-circulant Hermitian self-dual code, where $\mu \in \{1,\ww,\vv\}$.
Then there is a quaternary modified four $\mu$-circulant Hermitian self-dual code $C_{\mu}(A,B)$ 
such that $C \cong C_{\mu}(A,B)$ and the first 
nonzero coordinate of the first row of $A$ is $1$.
\end{lem}
% \begin{proof}
% The proofs of the assertions are similar for $\mu = 1,\ww,\vv$, and we give a proof for $\mu = \ww$.
% Suppose that $C$ is $C_{\ww}(A',B')$ and the first 
% nonzero coordinate of the first row of $A'$ is $\ww$ (resp.\ $\vv$).
% Then $C_{\ww}(\vv A', \vv B')$  (resp.\ $C_{\ww}(\ww A', \ww B')$) is a 
% modified four $\ww$-circulant code 
% such that nonzero coordinate of the first row of $\vv A'$ (resp.\ $\ww A'$)
% is $1$.
% By Lemma~\ref{lem:B:equiv} (i), we have that
% $C \cong C_{\ww}(\vv A', \vv B')$ (resp.\ $C \cong C_{\ww}(\ww A', \ww B')$).
% \end{proof}
\begin{proof}
Suppose that $C=C_{\mu}(A',B')$ and the first 
nonzero coordinate of the first row of $A'$ is $\ww$ (resp.\ $\vv$).
Then $C_{\mu}(\vv A', \vv B')$  (resp.\ $C_{\mu}(\ww A', \ww B')$) is a 
modified four $\mu$-circulant code 
such that nonzero coordinate of the first row of $\vv A'$ (resp.\ $\ww A'$)
is $1$.
By Lemma~\ref{lem:B:equiv} (i), we have that
$C \cong C_{\mu}(\vv A', \vv B')$ (resp.\ $C \cong C_{\mu}(\ww A', \ww B')$).
The result follows.
\end{proof}

The above lemma substantially reduces the number of codes which need be checked
when a classification of  modified  four $\mu$-circulant 
Hermitian self-dual codes is completed and the largest minimum weight among 
all modified  four $\mu$-circulant Hermitian self-dual codes is determined in the next section.

%%%%%%%%%%%%%%%%%%%%%%%%%%%%%
\section{Numerical results of modified four $\mu$-circulant Hermitian self-dual codes}
\label{Sec:N}

In this section, we present numerical results of 
quaternary modified  four $\mu$-circulant Hermitian self-dual codes.
We emphasize that Hermitian self-dual $[56,28,16]$ codes are constructed.
These codes are the first examples of not only
Hermitian self-dual $[56,28,16]$ codes but also (linear) $[56,28,16]$ codes.
All computer calculations in this section were done using programs in 
\textsc{Magma}~\cite{Magma} unless otherwise specified.
%   and \textsc{Mathematica}~\cite{Mathematica}.

%%%%
\subsection{Classification of modified  four $\mu$-circulant  Hermitian self-dual codes}

As described in Section~\ref{Sec:2},
all quaternary Hermitian self-dual codes of lengths up to $20$ were classified in~\cite{CPS},
\cite{HLMT},
\cite{HM11} and~\cite{MOSW}.
From now on, we consider Hermitian self-dual codes  for only lengths $n \ge 24$.

Let $d(n)$ denote  the largest minimum weight among all Hermitian self-dual codes of length $n$.
Let $d^K(n)$ denote  the largest minimum weight among previously known 
Hermitian self-dual codes of length $n$.
% For $n \in \{24,28,\ldots,80\}$, 
% the following values $d^K(n)$ are known:
% \begin{equation}\label{eq:dK}
% %d^K(n)=
% %\begin{cases}
% % 8& n =24, \\
% %10& n \in \{28,32\}, \\
% %12& n \in \{36,40,44\}, \\
% %14& n \in \{48,52,56\},\\
% %16& n \in \{60,64\},\\
% %18& n=68,
% %\end{cases}
% \begin{array}{c|c|c|c|c|c|c|c}
% n & 24 &28,32 &36,40,44 &48,52,56 &60,64 &68,72,76 & 80\\
%  \hline
% d^K(n)&   8& 10&  12& 14&  16&  18 & 20\\
% \end{array}
% \end{equation}
% where $d(24)=8$ and $d(28)=10$ (see~\cite[Table~5]{GG}).
For $n \in \{24,28,\ldots,80\}$, 
the values $d^K(n)$ are listed in Table~\ref{Tab:dK},
noting that $d(24)=8$ and $d(28)=10$ (see~\cite[Table~5]{GG}).

       %%%%%%%%%%%%%%%%%  table  %%%%%%%%%%%%%%%%%
\begin{table}[thb]
\caption{Values $d^K(n)$}
\label{Tab:dK}
\centering
\medskip
{\small
%{\footnotesize
%{\scriptsize
\begin{tabular}{c|c|c|c|c|c|c|c}
\noalign{\hrule height1pt}
$n$ & $24$ &$28,32$ &$36,40,44$ &$48,52,56$ &$60,64$ &$68,72,76$ & $80$\\
 \hline
$d^K(n)$&   $8$& $10$&  $12$& $14$&  $16$&  $18$ & $20$\\
\noalign{\hrule height1pt}
\end{tabular}
}
\end{table}
     %%%%%%%%%%%%%%%%%  table  %%%%%%%%%%%%%%%%%

Here we give a classification of 
modified four $\mu$-circulant Hermitian self-dual codes having minimum weight $d^K(n)$ 
for length $n\in \{24,28,32,36\}$.  
We describe how to complete our classification briefly.
Our exhaustive computer search based on Lemmas~\ref{lem:B:SD} and~\ref{lem:B:1}
found all distinct generator matrices~\eqref{eq:const} of 
modified four $\mu$-circulant Hermitian self-dual $[n,n/2,d^K(n)]$ codes $C_{\mu}(A,B)$,
which must be checked further for equivalences.
To test equivalence of two modified four $\mu$-circulant Hermitian self-dual $[n,n/2,d^K(n)]$ codes,
we used \textsc{Magma} function \texttt{IsIsomorphic}.
Moreover, in the process of finding these codes, 
we verified that there is no modified four $\mu$-circulant Hermitian self-dual 
code of length $n$ and minimum weight $d >d^K(n)$ for lengths $n=32$ and $36$. 
Then we have the following proposition.

\begin{prop}\label{prop:classification}
\begin{enumerate}
\item
Up to equivalence, 
there are $7$ quaternary modified four circulant Hermitian self-dual $[24,12,8]$ codes.
Up to equivalence, 
there are $9$  quaternary modified four $\mu$-circulant Hermitian self-dual $[24,12,8]$ codes
for $\mu \in \{\ww,\vv\}$.

\item
Up to equivalence, 
there are $3$ quaternary modified four $\mu$-circulant extremal Hermitian self-dual $[28,14,10]$ codes
for $\mu \in \{1,\ww,\vv\}$.

\item
Up to equivalence, 
there are $59$ quaternary modified four $\mu$-circulant Hermitian self-dual $[32,16,10]$ codes
for $\mu \in \{1,\ww,\vv\}$.
If $d \ge 12$, then there is no quaternary modified four $\mu$-circulant Hermitian self-dual $[32,16,d]$
code for $\mu \in \{1,\ww,\vv\}$.

\item
Up to equivalence, there is a unique  quaternary
modified four $\mu$-circulant Hermitian self-dual $[36,18,12]$ code
for $\mu \in \{1,\ww,\vv\}$.
If $d\ge 14$, then there is no quaternary modified four $\mu$-circulant Hermitian self-dual $[36,18,d]$
code for $\mu \in \{1,\ww,\vv\}$.
\end{enumerate}
\end{prop}

For $n \in \{24,28,32,36\}$,
by $C_{n,\mu,i}$ $(i \in \{1,2,\ldots,N_\mu(n)\})$,
we denote 
the modified four $\mu$-circulant Hermitian self-dual $[n,n/2,d^K(n)]$ codes
described in the above proposition,
where
\[
\begin{array}{ll}
N_{1}(24)=7, N_{\mu}(24)=9 \ (\mu \in \{\ww,\vv\}),
&
N_{\mu}(28)=3 \ (\mu \in \{1,\ww,\vv\}),
\\
N_{\mu}(32)=59 \ (\mu \in \{1,\ww,\vv\}),
&
N_{\mu}(36)=1 \ (\mu \in \{1,\ww,\vv\}).
\end{array}
%\begin{array}{ll}
%i_{1}(24)=7, i_{\ww}(24)=i_{\vv}(24)=9,&i_{1}(28)=i_{\ww}(28)=i_{\vv}(28)=3,\\
%i_{1}(32)=i_{\ww}(32)=i_{\vv}(32)=59,&i_{1}(36)=i_{\ww}(36)=i_{\vv}(36)=1.
%\end{array}
\]
For these codes $C_{n,\mu,i}=C_{\mu}(A,B)$ $(\mu =1,\ww,\vv)$,
the first rows $r_A$ (resp.\ $r_B$) of $A$ (resp.\ $B$) are
listed in Tables~\ref{Tab:24},
\ref{Tab:28},
\ref{Tab:36},
\ref{Tab:32-1}, 
\ref{Tab:32-2} and~\ref{Tab:32-3}.

\begin{rem}
By \textsc{Magma} function \texttt{IsIsomorphic}, we have the following
\[
\begin{array}{ll}
C_{24,\ww,i} \cong C_{24,\vv,i} & (i \in \{1,2,\ldots,9\}),\\
C_{28,1,i} \cong C_{28,\ww,i} \cong C_{28,\vv,i} & (i \in \{1,2,3\}),\\
C_{32,1,i} \cong C_{32,\ww,i} \cong C_{32,\vv,i} & (i \in \{1,2,\ldots,59\}),\\
C_{36,\ww,1} \cong C_{36,\vv,1},
\end{array}
\]
and there is no other pair of equivalent codes among the codes described in 
Proposition~\ref{prop:classification}.
\end{rem}

       %%%%%%%%%%%%%%%%%  table  %%%%%%%%%%%%%%%%%
\begin{table}[thbp]
\caption{Modified four $\mu$-circulant Hermitian self-dual $[24,12,8]$ codes}
\label{Tab:24}
\centering
\medskip
{\small
%{\footnotesize
%{\scriptsize
\begin{tabular}{c|ll|c}
\noalign{\hrule height1pt}
Code & \multicolumn{1}{c}{$r_A$} & \multicolumn{1}{c|}{$r_B$} & $A_{8}$   \\
\hline
$C_{24,1,1}$ &$(1,0,1,0,1,1)$&$(\ww,\vv,\ww,1,0,1)$&513 \\
$C_{24,1,2}$ &$(1,\ww,1,1,\ww,1)$&$(\ww,1,\ww,1,0,1)$&594 \\
$C_{24,1,3}$ &$(1,\vv,\vv,0,\ww,1)$&$(\vv,0,0,\vv,0,0)$&594 \\
$C_{24,1,4}$ &$(0,1,1,0,\ww,\ww)$&$(\ww,\vv,\ww,1,0,1)$&837 \\
$C_{24,1,5}$ &$(0,1,1,0,\ww,\ww)$&$(1,1,1,0,0,0)$&837 \\
$C_{24,1,6}$ &$(1,1,1,0,\ww,\vv)$&$(1,\ww,1,1,0,0)$&837 \\
$C_{24,1,7}$ &$(1,\vv,\vv,\vv,1,\vv)$&$(\vv,\vv,\vv,0,0,0)$&837 \\
$C_{24,\ww,1}$ &$(1,\vv,\vv,1,\vv,\ww)$&$(\vv,0,0,\vv,1,0)$&513 \\
$C_{24,\ww,2}$ &$(1,1,1,\ww,\vv,0)$&$(\vv,0,1,\vv,\ww,0)$&513 \\
$C_{24,\ww,3}$ &$(1,\vv,\vv,1,\vv,0)$&$(\ww,0,\vv,\vv,1,0)$&513 \\
$C_{24,\ww,4}$ &$(1,\vv,\vv,1,\ww,0)$&$(0,0,0,\vv,1,0)$&513 \\
$C_{24,\ww,5}$ &$(1,1,1,\ww,0,\vv)$&$(1,\ww,\vv,1,\ww,\ww)$&513 \\
$C_{24,\ww,6}$ &$(1,\vv,\vv,1,1,\vv)$&$(\ww,\ww,\vv,\ww,1,0)$&513 \\
$C_{24,\ww,7}$ &$(1,\vv,\vv,1,\ww,1)$&$(0,0,\ww,\ww,1,0)$&513 \\
$C_{24,\ww,8}$ &$(0,1,1,\ww,\vv,\vv)$&$(\ww,0,1,\vv,\ww,0)$&513 \\
$C_{24,\ww,9}$ &$(0,1,\ww,\vv,0,0)$&$(1,\vv,0,1,\vv,0)$&513 \\
$C_{24,\vv,1}$ &$(0,0,1,1,0,\ww)$&$(\vv,\vv,1,\vv,\vv,\ww)$&513 \\
$C_{24,\vv,2}$ &$(1,0,\vv,\vv,\vv,1)$&$(0,\ww,0,\vv,\ww,1)$&513 \\
$C_{24,\vv,3}$ &$(1,0,\vv,\vv,1,0)$&$(0,\ww,1,\ww,\ww,1)$&513 \\
$C_{24,\vv,4}$ &$(0,0,1,1,0,0)$&$(\ww,\ww,\ww,0,\vv,\ww)$&513 \\
$C_{24,\vv,5}$ &$(0,0,1,1,0,0)$&$(\ww,0,1,\ww,\vv,\ww)$&513 \\
$C_{24,\vv,6}$ &$(0,0,1,1,1,1)$&$(\vv,0,\vv,1,\vv,\ww)$&513 \\
$C_{24,\vv,7}$ &$(0,0,1,1,\vv,0)$&$(\ww,\vv,\vv,\vv,\vv,\ww)$&513 \\
$C_{24,\vv,8}$ &$(0,0,1,1,1,0)$&$(\vv,\ww,\vv,1,\vv,\ww)$&513 \\
$C_{24,\vv,9}$ &$(0,0,1,1,1,0)$&$(0,\ww,1,0,\vv,\ww)$&513 \\
\noalign{\hrule height1pt}
\end{tabular}
}
\end{table}
     %%%%%%%%%%%%%%%%%  table  %%%%%%%%%%%%%%%%%

       %%%%%%%%%%%%%%%%%  table  %%%%%%%%%%%%%%%%%
\begin{table}[thbp]
\caption{Modified four $\mu$-circulant Hermitian self-dual $[28,14,10]$ codes}
\label{Tab:28}
\centering
\medskip
{\small
%{\footnotesize
%{\scriptsize
\begin{tabular}{c|ll}
\noalign{\hrule height1pt}
Code & \multicolumn{1}{c}{$r_A$} & \multicolumn{1}{c}{$r_B$}  \\
\hline
$C_{28,1,1}$ & $(0,1,\vv,0,1,0,\vv)$ & $(\vv,\ww,\vv,\ww,0,0,\vv)$ \\
$C_{28,1,2}$ & $(1,\ww,1,0,\ww,0,0)$ & $(\vv,\vv,1,\ww,\ww,1,1)$ \\
$C_{28,1,3}$ & $(0,1,\ww,0,\vv,0,0)$ & $(\ww,1,\ww,1,1,0,\ww)$ \\
$C_{28,\ww,1}$ & $(1,\ww,\ww,0,\ww,0,1)$ & $(\ww,0,0,\vv,\ww,0,\ww)$ \\
$C_{28,\ww,2}$ & $(1,\ww,\ww,0,\ww,0,0)$ & $(1,\vv,\vv,\ww,\ww,1,\ww)$ \\
$C_{28,\ww,3}$ & $(0,1,\ww,0,\ww,0,0)$ & $(\vv,\ww,0,\ww,1,1,\ww)$ \\
$C_{28,\vv,1}$ & $(1,0,\vv,\ww,\vv,0,1)$ & $(\ww,0,0,0,1,1,\vv)$ \\
$C_{28,\vv,2}$ & $(1,\ww,1,\vv,1,1,\ww)$ & $(\vv,0,\vv,0,0,\ww,1)$ \\
$C_{28,\vv,3}$ & $(1,\vv,\vv,\ww,\vv,0,1)$ & $(0,\vv,\vv,0,0,0,\vv)$ \\
\noalign{\hrule height1pt}
\end{tabular}
}
\end{table}
     %%%%%%%%%%%%%%%%%  table  %%%%%%%%%%%%%%%%%

For $n=24,32$ and $36$, the possible weight enumerators 
of quaternary Hermitian  self-dual $[n,n/2,d^K(n)]$ 
codes can be written using $A_{d^K(n)}$ (see~\cite{AH} and~\cite[Section~III]{K01}).
Note that the possible weight enumerator
of an extremal Hermitian  self-dual code of a fixed length is uniquely determined.
For the above codes $C_{n,\mu,i}$ ($n=24,32$ and $36$),
the numbers $A_{d^K(n)}$ of codewords of minimum weight $d^K(n)$
are also listed in 
Tables~\ref{Tab:24},
\ref{Tab:36},
\ref{Tab:32-1}, 
\ref{Tab:32-2} and~\ref{Tab:32-3}.
This was calculated by the \textsc{Magma} function \texttt{NumberOfWords}.

       %%%%%%%%%%%%%%%%%  table  %%%%%%%%%%%%%%%%%
\begin{table}[thbp]
\caption{Modified four $\mu$-circulant Hermitian self-dual $[36,18,12]$ codes}
\label{Tab:36}
\centering
\medskip
{\small
%{\footnotesize
%{\scriptsize
\begin{tabular}{c|ll|c}
\noalign{\hrule height1pt}
Code & \multicolumn{1}{c}{$r_A$} & \multicolumn{1}{c|}{$r_B$} & $A_{12}$  \\
\hline
$C_{36,1,1}$&$(1,1,1,1,1,1,\ww,\vv,\ww)$&$(1,1,0,0,1,0,1,0,0)$&20844\\
$C_{36,\ww,1}$&$(1,\ww,1,1,1,1,\vv,1,1)$&$(1,\ww,\vv,\vv,0,\vv,1,0,0)$&19548\\
$C_{36,\vv,1}$&$(1,\vv,1,1,1,1,\ww,1,1)$&$(1,\vv,\ww,\ww,0,\ww,1,0,0)$&19548\\
\noalign{\hrule height1pt}
\end{tabular}
}
\end{table}
     %%%%%%%%%%%%%%%%%  table  %%%%%%%%%%%%%%%%%

%In addition, by \textsc{Magma} function \texttt{IsIsomorphic}, we have the following:
%
%\begin{rem}
%For each $i \in \{1,2,\ldots,9\}$,
%the two codes $C_{24,\ww,i}$ and  $C_{24,\vv,i}$ are equivalent.
%For each $i \in \{1,2,3\}$,
%the three codes $C_{28,1,i}$,  $C_{28,\ww,i}$ and  $C_{28,\vv,i}$ are equivalent.
%For each $i \in \{1,2,\ldots,59\}$,
%the three codes $C_{32,1,i}$,  $C_{32,\ww,i}$ and  $C_{32,\vv,i}$ are equivalent.
%The two codes $C_{36,\ww}$ and $C_{36,\vv}$ are equivalent.
%\end{rem}

%%%%
\subsection{Largest minimum weights of modified  four $\mu$-circulant Hermitian self-dual codes}

We give some observations on 
the largest minimum weight $d(n)$ among all Hermitian self-dual codes of length $n$
and the largest minimum weight
$d_{\mu}(n)$ $(\mu =1,\ww,\vv)$ among all
modified four $\mu$-circulant Hermitian self-dual codes of length $n$.
%Let $d_4(n)$ denote the value $\max\{d_1(n), d_\ww(n), d_{\vv}(n)\}$.
For lengths $n=40$ and $44$, 
by a method similar to the above, 
our exhaustive computer search based on Lemmas~\ref{lem:B:SD} and~\ref{lem:B:1}
verified that 
there is no modified four $\mu$-circulant Hermitian self-dual $[n,n/2,d]$
code with $d > d^K(n)$ for $\mu \in \{1,\ww,\vv\}$
(see Table~\ref{Tab:dK} for the minimum weights $d^K(n)$).
In addition, we found a
modified four $\mu$-circulant Hermitian self-dual $[n,n/2,d^K(n)]$
code $C_{n,\mu}$ for $\mu \in \{1,\ww,\vv\}$.
This implies the following proposition.

\begin{prop}
For $\mu \in \{1,\ww,\vv\}$,
$d_{\mu}(40)=10\text{ and }
d_{\mu}(44)=12$.
\end{prop}

%Note that the possible weight enumerators 
%of quaternary Hermitian  self-dual $[40,20,12]$ 
%codes are given in~\cite[Section~3.7]{GHM}.
For the above codes $C_{40,\mu}=C_{\mu}(A,B)$ and $C_{44,\mu}=C_{\mu}(A,B)$,
the first rows $r_A$ (resp.\ $r_B$) of $A$ (resp.\ $B$) are
listed in Table~\ref{Tab:4044}.
The numbers $A_{d^K(n)}$ of codewords of minimum weight $d^K(n)$
are also listed in the table.
This was calculated by the \textsc{Magma} function \texttt{NumberOfWords}.
The numbers show that these codes are inequivalent.

       %%%%%%%%%%%%%%%%%  table  %%%%%%%%%%%%%%%%%
\begin{table}[thb]
\caption{Modified four $\mu$-circulant Hermitian self-dual codes $C_{40,\mu}$ and $C_{44,\mu}$}
\label{Tab:4044}
\centering
\medskip
{\small
%{\footnotesize
%{\scriptsize
\begin{tabular}{c|ll|c}
\noalign{\hrule height1pt}
Code & \multicolumn{1}{c}{$r_A$} & \multicolumn{1}{c|}{$r_B$} &$A_{d^K(n)}$  \\
\hline
$C_{40,1}$   &$(1,0,0,1,\vv,1,0,0,1,0)$&$(\ww,\ww,1,1,\ww,\ww,0,\ww,\ww,0)$ & $5220$ \\
$C_{40,\ww}$ &$(1,\vv,\vv,1,1,1,\ww,\vv,\vv,\vv)$&$(\ww,1,0,\ww,\vv,\ww,\ww,\vv,0,0)$ & $5130$ \\
$C_{40,\vv}$ &$(1,\ww,\ww,\ww,\vv,\vv,\ww,\vv,0,\ww)$&$(1,\ww,0,1,\vv,0,\vv,0,\ww,0)$ & $5040$ \\
$C_{44,1}$   &$(1,\ww,\vv,0,0,\ww,0,\vv,\ww,1,\ww)$&$(\vv,\ww,\vv,\ww,\ww,\vv,0,\ww,0,0,0)$ & $1188$ \\
$C_{44,\ww}$ &$(1,0,0,\vv,\ww,1,\vv,0,\ww,1,0)$&$(\vv,0,\vv,0,\ww,0,0,\vv,1,\ww,0)$ & $1551$ \\
$C_{44,\vv}$ &$(1,0,\vv,0,1,1,1,\ww,0,\vv,\ww)$&$(\ww,1,\vv,\vv,0,\vv,\vv,\ww,\ww,\vv,0)$ & $1749$ \\
\noalign{\hrule height1pt}
\end{tabular}
}
\end{table}
     %%%%%%%%%%%%%%%%%  table  %%%%%%%%%%%%%%%%%

%For lengths $n \in \{4k \mid k \in \{12,13,\ldots,20\}\}$,
%For lengths $48,52,56,60,64,68,72,76$ and $80$, 
For lengths $48,52,\ldots,76$ and $80$, 
by a non-exhaustive search  based on Lemmas~\ref{lem:B:SD} and~\ref{lem:B:1}, 
we  continued finding
modified four $\mu$-circulant Hermitian self-dual codes having large minimum
weights.
Then we found a
modified four $\mu$-circulant Hermitian self-dual code $C_{n,\mu}$ 
of length $n$ and minimum weight $d$ for
\begin{align*}
(n,\mu,d)=&
(52, 1, 14),
(52, \ww, 14),
(52, \vv, 14),
%\\&
(56, 1, 16),
(56, \ww, 16),
(56, \vv, 14),
\\&
(60, 1, 16),
(60, \ww, 16),
(60, \vv, 16),
%\\&
(64, 1, 16),
(64, \ww, 16),
(64, \vv, 16),
\\&
(68, 1, 18),
(68, \ww, 18),
(68, \vv, 18),
%\\&
(72, 1, 18),
(72, \ww, 18),
(72, \vv, 18),
\\&
(76, 1, 18),
(76, \ww, 18),
(76, \vv, 18),
%\\&
(80, 1, 20),
(80, \ww, 20),
(80, \vv, 20).
\end{align*}
For the above codes $C_{n,\mu}=C_{\mu}(A,B)$,
the first rows $r_A$ (resp.\ $r_B$) of $A$ (resp.\ $B$) are
listed in Table~\ref{Tab:L}.
%In the process of constructing these codes, 
%From Table~\ref{Tab:L}, we have the following:
We have the following proposition.

\begin{prop}\label{prop:56}
There are quaternary Hermitian self-dual $[56,28,16]$ codes.
\end{prop}

We emphasize that $C_{56,1}$ and $C_{56,\ww}$ are the first examples of not only
Hermitian self-dual $[56,28,16]$ codes but also (linear) $[56,28,16]$ codes~\cite{Grassl}.
We give the weight enumerators of  these codes in the next subsection.

In Table~\ref{Tab:d}, we summarize the current information on
$d_1(n)$, $d_\ww(n)$  and $d_{\vv}(n)$.
The upper bounds on $d_1(n)$, $d_\ww(n)$ and $d_{\vv}(n)$ follow from~\eqref{eq:bound}.
The lower bounds on $d_1(n)$, $d_\ww(n)$ and $d_{\vv}(n)$ follow from 
Table~\ref{Tab:L}.

       %%%%%%%%%%%%%%%%%  table  %%%%%%%%%%%%%%%%%
\begin{table}[th]
\caption{Largest minimum weights $d_1(n)$, $d_\ww(n)$ and $d_{\vv}(n)$}
\label{Tab:d}
\centering
\medskip
{\small
%{\footnotesize
%{\scriptsize
\begin{tabular}{c|ccc||c|ccc}
\noalign{\hrule height1pt}
$n$ & $d_1(n)$ & $d_\ww(n)$ & $d_{\vv}(n)$ &
$n$ & $d_1(n)$ & $d_\ww(n)$ & $d_{\vv}(n)$ \\
\hline
24 &8 &8 &8 &56 & {16}--20 & {16}--20& 14--20\\
28 &10 &10 &10 &60 & 16--22 & 16--22& 16--22\\
32 &10 &10 &10 &64 & 16--22& 16--22& 16--22\\
36 &12 &12 &12 &68 & 18--24& 18--24& 18--24\\
40 &12 &12 &12 &72 & 18--26& 18--26& 18--26\\
44 &12 &12 &12 &76 & 18--26& 18--26& 18--26\\
48 & 14--18& 14--18 & 14--18   &80 & 20--28& 20--28& 20--28\\
52 & 14--18& 14--18& 14--18 & & & & \\
\noalign{\hrule height1pt}
\end{tabular}
}
\end{table}
     %%%%%%%%%%%%%%%%%  table  %%%%%%%%%%%%%%%%%

%        %%%%%%%%%%%%%%%%%  table  %%%%%%%%%%%%%%%%%
% \begin{table}[th]
% \caption{Largest minimum weights $d_1(n)$, $d_\ww(n)$, $d_{\vv}(n)$  and $d(n)$}
% \label{Tab:d}
% \centering
% \medskip
% {\small
% %{\footnotesize
% %{\scriptsize
% \begin{tabular}{c|ccc|c|c}
% \noalign{\hrule height1pt}
% $n$ & $d_1(n)$ & $d_\ww(n)$ & $d_{\vv}(n)$ & $d(n)$ & References \\
% \hline
% 24 &8 &8 &8 &8 &  (see~\cite[p.~140]{CP}) \\
% 28 &10 &10 &10 &10 &  (see~\cite[Theorem~9]{Hu90})\\
% %
% 32 &10 &10 &10 &10--12 &\cite[Table~I]{G00} \\
% 36 &12 &12 &12 &12--14 &\cite[Table~I]{G00} \\
% 40 &12 &12 &12 &12--14 &\cite[Table~I]{G00} \\
% 44 &12 &12 &12 &12--16 &\cite[Table~5]{GG} \\
% 48 & 14--18& 14--18 & 14--18 &  14--18  &\cite[Table~5]{GG} \\
% 52 & 14--18& 14--18& 14--18&  14--18  &\cite[Table~5]{GG} \\
% 56 & {\bf 16}--20 & {\bf 16}--20& 14--20&  {\bf 16}--20 & $C_{56,1},C_{56,\ww}$ in Table~\ref{Tab:L} \\
% 60 & 16--22 & 16--22& 16--22&  16--22  &\cite[Table~5]{GG} \\
% 64 & 16--22& 16--22& 16--22&  16--22  &\cite[Table~5]{GG} \\
% 68 & 18--24& 18--24& 18--24&  18--24  &\cite[Table~5]{GG} \\
% 72 & 18--26& 18--26& 18--26&  18--26  &\cite[Table~5]{GG} \\
% 76 & 18--26& 18--26& 18--26&  18--26  &\cite[Table~5]{GG} \\
% 80 & 20--28& 20--28& 20--28&  20--28  &\cite[Table~5]{GG} \\
% \noalign{\hrule height1pt}
% \end{tabular}
% }
% \end{table}
%      %%%%%%%%%%%%%%%%%  table  %%%%%%%%%%%%%%%%%

%%%%
\subsection{$C_{56,1}$ and $C_{56,\ww}$}

It is a main problem in coding theory to determine
the largest minimum weight $d_q(n,k)$ among
all $[n,k]$ codes over a finite field of order $q$
for a given $(q,n,k)$.
The current information on $d_4(n,k)$ can be found in~\cite{Grassl}.
For example, it was previously known that
$15 \le d_4(56,28) \le 21$.
As a consequence of Proposition~\ref{prop:56}, we have the following corollary.

\begin{cor}\label{cor:56}
% There are quaternary $[56,28,16]$ codes, and 
$16 \le d_4(56,28) \le 21$.
\end{cor}

Now we determine the weight enumerators of $C_{56,1}$ and $C_{56,\ww}$.
It is well known that the possible weight enumerators of  quaternary Hermitian  self-dual codes
can be determined by the Gleason type theorem~\cite[p.~804]{MMS}
(see also~\cite[Theorem~13]{MOSW}). 
The weight enumerator $W$ of a  quaternary Hermitian  self-dual code of length $n$
is written as:
\begin{equation}\label{eq:F4:G}
{W = \sum_{j=0}^{\lfloor \frac{n}{6} \rfloor} 
a_j(1+3y^2)^{\frac{n}{2}-3j}(y^2(1-y^2)^2)^j,}  
\end{equation}
using some integers $a_j$.
The possible weight enumerator $W_{56,16}=\sum_{i=0}^{56}A_i y^i$
of a quaternary Hermitian  self-dual $[56,28,16]$ 
code is determined by~\eqref{eq:F4:G},
where $A_i$ are listed in Table~\ref{Tab:W56} together with
$\alpha=A_{16}$ and $\beta=A_{18}$.
Only this calculation was done by \textsc{Mathematica}~\cite{Mathematica}.
By the \textsc{Magma} function \texttt{NumberOfWords},
we calculated that 
\begin{align*}
(A_{16},A_{18})= &
(48825,2275560)
\text{ and }
(47544, 2282700),
\end{align*}
for $C_{56,1}$ and $C_{56,\ww}$, respectively.
This determines the weight enumerators of $C_{56,1}$ and $C_{56,\ww}$.

       %%%%%%%%%%%%%%%%%  table  %%%%%%%%%%%%%%%%%
\begin{table}[th]
\caption{Possible weight enumerator $W_{56,16}$}
\label{Tab:W56}
\centering
\medskip
{\small
%{\footnotesize
%{\scriptsize
\begin{tabular}{c|r}
\noalign{\hrule height1pt}
$i$ & \multicolumn{1}{c}{$A_i$} \\
\hline
 0&$1$ \\
16&$\alpha$ \\
18&$\beta$ \\
20&$ 113963850 - 78 \alpha - 15 \beta$ \\
22&$ 1616214600 + 520 \alpha + 99 \beta$ \\
24&$ 35022262275 - 1495 \alpha - 357 \beta$ \\
26&$ 467452738368 + 1344 \alpha + 612 \beta$ \\
28&$ 4854958425240 + 5560 \alpha + 612 \beta$ \\
30&$ 37999586848608 - 28576 \alpha - 7140 \beta$ \\
32&$ 223928221341825 + 79170 \alpha + 23868 \beta$ \\
34&$ 991894905892800 - 170560 \alpha - 51714 \beta$ \\
36&$ 3272633909885340 + 309452 \alpha + 82654 \beta$ \\
38&$ 7961209635178800 - 471120 \alpha - 102102 \beta$ \\
40&$ 14053893738878070 + 586586 \alpha + 99450 \beta$ \\
42&$ 17629097730552000 - 584000 \alpha - 76908 \beta$ \\
44&$ 15262097167863000 + 457080 \alpha + 47124 \beta$ \\
46&$ 8759255147042400 - 276640 \alpha - 22644 \beta$ \\
48&$ 3144896807802750 + 126685 \alpha + 8364 \beta$ \\
50&$ 646962821144640 - 42432 \alpha - 2295 \beta$ \\
52&$ 65864956983210 + 9810 \alpha + 441 \beta$ \\
54&$ 2485731965640 - 1400 \alpha - 53 \beta$ \\
56&$ 14512944519 + 93 \alpha + 3 \beta$ \\
\noalign{\hrule height1pt}
\end{tabular}
}
\end{table}
     %%%%%%%%%%%%%%%%%  table  %%%%%%%%%%%%%%%%%

%From Table~\ref{Tab:W56}, $C_{56,\ww}$ has the following weight enumerator:
%\begin{align*}
%&1+ 47544 y^{16}+ 2282700 y^{18}+ 76014918 y^{20}+ 1866924780 y^{22}\\&
%+ 34136260095 y^{24}+ 468913649904 y^{26}+ 4856619782280 y^{28}\\&
%+ 37981929753264 y^{30}+ 223986468883905 y^{32}+ 991768749240360 y^{34}\\&
%+ 3272837296757028 y^{36}+ 7960954168014120 y^{38}\\&
%+ 14054148642037854 y^{40}+ 17628894406964400 y^{42}\\&
%+ 15262226469229320 y^{44}+ 8759190305011440 y^{46}\\&
%+ 3144921923417190 y^{48}+ 646955564961132 y^{50}\\&
%+ 65866430060550 y^{52}+ 2485544420940 y^{54}+ 14524214211 y^{56}.
%\end{align*}

%%%%%%%%%%%%%%%%%%%%%%%%%
\subsection{Largest minimum weights $d(n)$}

In Table~\ref{Tab:dmax},
we summarize the current information on the largest minimum weights $d(n)$ for $n \in \{24,28,\ldots,80\}$.
The upper bounds on $d(n)$ follow from~\eqref{eq:bound}.
The references about the lower bounds on $d(n)$ are also listed in the table.

       %%%%%%%%%%%%%%%%%  table  %%%%%%%%%%%%%%%%%
\begin{table}[th]
\caption{Largest minimum weights $d(n)$}
\label{Tab:dmax}
\centering
\medskip
{\small
%{\footnotesize
%{\scriptsize
\begin{tabular}{c|cc||c|cc}
\noalign{\hrule height1pt}
$n$ & $d(n)$ & References &$n$ & $d(n)$ & References \\
\hline
24 &8 &  (see~\cite[p.~140]{CP}) &64 &16--22  &\cite[Table~5]{GG} \\
28 &10 &  (see~\cite[Theorem~9]{Hu90})&68 &18--24  &\cite[Table~5]{GG} \\
32 &10--12 &\cite[Table~I]{G00} &72 &18--26  &\cite[Table~5]{GG} \\
36 &12--14 &\cite[Table~I]{G00} &76 &18--26  &\cite[Table~5]{GG} \\
40 &12--14 &\cite[Table~I]{G00} &80 &20--28  &\cite[Table~5]{GG} \\
44 &12--16 &\cite[Table~5]{GG} &84 &{\bf 20}--30  & $C_{84,\ww}$ in Table~\ref{Tab:L}\\
48 &14--18  &\cite[Table~5]{GG} &88 &{\bf 20}--30  &$C_{88,\ww}$ in Table~\ref{Tab:L}\\
52 &14--18  &\cite[Table~5]{GG} &92 &22--32  &$G_{92}$ (see~\cite{Grassl}) \\
56 &{\bf 16}--20 & $C_{56,1},C_{56,\ww}$ in Table~\ref{Tab:L} 
      &96 &{\bf 22}--34 & $C_{96,\ww}$ in Table~\ref{Tab:L}\\
60 &16--22  &\cite[Table~5]{GG} &100 &22--34  & $G_{100}$ (see~\cite{Grassl}) \\
\noalign{\hrule height1pt}
\end{tabular}
}
\end{table}
     %%%%%%%%%%%%%%%%%  table  %%%%%%%%%%%%%%%%%

In \cite[Table~5]{GG}, the largest minimum weights $d(n)$ were considered for $n \le 80$.
Here we investigate  the largest minimum weights $d(n)$ for $n \in \{84, 88, 92, 96, 100\}$.
A Hermitian self-dual code of length $n$ and minimum weight $22$ 
is given  in \cite{Grassl} for $n=92$ and $100$.
We denote the two codes
by $G_{92}$ and $G_{100}$, respectively.
As information, we briefly give the construction of $G_{92}$ and $G_{100}$.
Let $G_{91,1}$ and $G_{91,2}$ denote the cyclic codes of length $91$ with generator
polynomials $g_1$ and $g_2$, respectively, where
\begin{align*}
g_1=&
x^{46} + \ww x^{44} + x^{43} + \vv x^{42} + \vv x^{41} + \ww x^{40}
+ \vv x^{39} + \ww x^{38} + x^{37} + x^{35} 
\\&
+ x^{34} + \ww x^{33}
+ x^{31} + \ww x^{30} + \vv x^{28} + x^{27} + x^{26} + \ww x^{25}
+ \ww x^{24} + x^{21}
\\&
 + x^{20} + x^{19} + x^{18} + \vv x^{16}
+ \ww x^{15} + \vv x^{14} + \vv x^{13} + \vv x^{12} + \vv x^{10}
\\&
+ \vv x^{8} +\vv x^{7} + x^{6} + \ww x^{5} + x^{4} + \ww x^{3}
+ \vv x^{2} + x +1,
\\
g_2=&
x^{45} + x^{44} + \vv x^{43} + \ww x^{42} + x^{41} + \ww x^{40}
+ \vv x^{38} + x^{37} + x^{34} + \ww x^{32} 
\\&
+ \ww x^{31} + \vv x^{30}
+ x^{29} + x^{28} + \ww x^{27} + \vv x^{26} + \ww x^{25} + \ww x^{23}
+ \ww x^{22} 
\\&
+ \ww x^{21} + \vv x^{20} + \ww x^{19} + \vv x^{18}
+ \ww x^{17} + \ww x^{16} + x^{15} + \vv x^{14} 
\\&
+ \vv x^{12} + \vv x^{9}
+ \vv x^{8} + \vv x^{6} + \ww x^{5} + x^{3} + \vv x^{2} + 1.
\end{align*}
The code $G_{92}$ is constructed from $G_{91,1}$, $G_{91,2}$ and the $[1,1]$ code
by Construction~X.
The code $G_{100}$ is equivalent to the double circulant code with generator matrix
$
\left(
\begin{array}{cc}
I_{50} & R 
\end{array}
\right)$,
where $R$ is the circulant matrix with the first row
\begin{multline*}
(\ww,\ww,1,\vv,\ww,1,1,\ww,\vv,0,1,0,\vv,1,\ww,\vv,\vv,1,\vv,\vv,0,1,\vv,
\vv,0,
\\
\ww,0,0,
1,1,\ww,\ww,\ww,1,\ww,\vv,1,1,1,1,\ww,1,\vv,1,\ww,\ww,\ww,
0,\vv,1).
\end{multline*}
For
$(n,d)=(84,20)$, $(88,20)$ \text{ and } $(96,22)$,
by a non-exhaustive search  based on Lemmas~\ref{lem:B:SD} and~\ref{lem:B:1}, 
we found a
modified four $\ww$-circulant Hermitian self-dual code $C_{n,\ww}=C_{\ww}(A,B)$ 
of length $n$ and minimum weight $d$.
For the above codes,
the first rows $r_A$ (resp.\ $r_B$) of $A$ (resp.\ $B$) are
listed in Table~\ref{Tab:L}.
In Table~\ref{Tab:d},
we give lower and upper bounds on the largest minimum weights $d(n)$ for 
$n \in \{84, 88, 92, 96, 100\}$.
The upper bounds on $d(n)$ follow from~\eqref{eq:bound}.
The references about the lower bounds on $d(n)$ are also listed in the table.
For 
\[
(n,d)=(56,16), (84,20), (88,20) \text{ and } (92,22),
\]
a Hermitian self-dual code of length $n$ and minimum weight $d$ is constructed
for the first time.
In Table~\ref{Tab:dmax}, the minimum weights of these codes are given in bold.

%%%%%%%%%%%%%%%%%%%%%%%%%
\section{Application to quantum codes}\label{Sec:Q}

In this section, we consider an application of  the quaternary Hermitian  self-dual $[56,28,16]$ 
codes $C_{56,1}$ and $C_{56,\ww}$ found in the previous section  to quantum codes.  

% A quaternary {\em additive} code $C$ of length $n$
% is an additive subgroup of $\FF_4^n$.
% Such a code is an $\FF_2$-subspace of $\FF_4^n$
% A vector of $C$ is called a {\em codeword} of $C$.
A quaternary \emph{additive} $(n,2^k)$ code $\cC$ is 
an additive subgroup of $\FF_4^n$ with $|\cC|=2^k$.
%The \emph{minimum weight} of $\cC$ is 
%the minimum non-zero weight of vectors of $\cC$.
%A quaternary additive $(n,2^k,d)$ code is 
%a quaternary additive $(n,2^k)$ code with minimum weight $d$. 
The {\em dual} code ${\cC}^*$ of a quaternary additive $(n,2^k)$ code $\cC$ 
is defined as
\[
{\cC}^*=
\{x \in \FF_4^n \mid x * y = 0 \text{ for all } y \in \cC\},
\]
under the following trace inner product
\[
x * y=\sum_{i=1}^n (x_iy_i^2+x_i^2y_i)
\]
for $x=(x_1,x_2,\ldots,x_n)$, $y=(y_1,y_2,\ldots,y_n) \in
\FF_4^n$.
A quaternary  additive code $\cC$ is called 
\emph{self-orthogonal} and 
\emph{self-dual} if $\cC \subset {\cC}^*$ and $\cC = {\cC}^*$, respectively.
Note that a  quaternary Hermitian self-dual $[n,n/2,d]$ code is a quaternary  additive self-dual
$(n,2^n)$ code with minimum weight $d$ (see e.g.,~\cite{GK}).

A useful method for constructing quantum codes from quaternary additive self-orthogonal 
 codes was given by Calderbank, Rains, Shor and 
Sloane~\cite{CRSS} (see~\cite{CRSS} for 
undefined terms concerning quantum codes).
A quaternary additive self-orthogonal $(n,2^{n-k})$ code $\cC$
such that there is no vector of weight less than $d$ in ${\cC}^* \setminus \cC$,
gives a quantum $[[n,k,d]]$ code, where $k \ne 0$.
A quaternary additive self-dual $(n,2^n)$ code with minimum weight $d$ gives a quantum $[[n,0,d]]$ code.
Let $d_{\max}(n,k)$ denote the largest minimum weight $d$
among quantum $[[n,k,d]]$ codes.
Similar to the classical coding theory,
it is a fundamental problem to determine $d_{\max}(n,k)$.
A table on $d_{\max}(n,k)$ is given in~\cite[Table~III]{CRSS} for $n \le 30$.
An extended table is obtained  electronically from~\cite{Grassl}.
For example, it was previously known that $15 \le d_{\max}(56,0) \le 20$~\cite{Grassl}.

In Section~\ref{Sec:N}, quaternary Hermitian self-dual $[56,28,16]$ codes $C_{56,1}$
and $C_{56,\ww}$ were constructed for the first time.
By the above method, a quantum $[[56,0,16]]$ code is obtained.
Hence, we have the following proposition.

\begin{prop}\label{prop:Q}
\begin{enumerate}
\item
There is a quantum $[[56,0,16]]$ code.
\item
$16 \le d_{\max}(56,0) \le 20$.
\end{enumerate}
\end{prop}

%The above quantum code improves the previously known lower bound on  $d_{\max}(56,0)$.

%%%%%%%%%%%%%%%%%%%%%%%%%%%%%
\bigskip
\noindent
\textbf{Acknowledgments.}
This work was supported by JSPS KAKENHI Grant Number 19H01802.

%%%%%%%%%%%%%%%%%%%  References  %%%%%%%%%%%%%%%%%%%%%%%%

%%%%%%%%%%%%%%%%%%%%%%%%%%%%%%%%%

\begin{landscape}

       %%%%%%%%%%%%%%%%%  table  %%%%%%%%%%%%%%%%%
\begin{table}[thbp]
\caption{Modified four circulant Hermitian self-dual $[32,16,10]$ codes $C_{32,1,i}$}
\label{Tab:32-1}
\centering
\medskip
%{\small
{\footnotesize
%{\scriptsize
\begin{tabular}{c|ll|c||c|ll|c}
\noalign{\hrule height1pt}
$i$ & \multicolumn{1}{c}{$r_A$} & \multicolumn{1}{c|}{$r_B$} & $A_{10}$&
$i$ & \multicolumn{1}{c}{$r_A$} & \multicolumn{1}{c|}{$r_B$}  & $A_{10}$ \\
\hline
${1}$&$(0,1,0,0,0,\vv,\vv,\ww)$&$(\ww,\vv,1,\ww,1,0,1,\ww)$&1200&
${31}$&$(0,1,1,0,0,1,\vv,\ww)$&$(\ww,\ww,1,0,\vv,0,1,\ww)$&1776\\
${2}$&$(0,1,1,0,0,\vv,\ww,\ww)$&$(1,\vv,\ww,\vv,0,0,1,\ww)$&1200&
${32}$&$(0,1,0,0,0,\vv,\ww,0)$&$(\vv,\ww,\ww,\ww,0,0,1,\ww)$&1776\\
${3}$&$(0,1,1,0,0,\vv,\ww,\ww)$&$(\ww,\vv,0,\vv,1,0,1,\ww)$&1200&
${33}$&$(0,1,0,0,0,1,\vv,\vv)$&$(\vv,1,1,\vv,1,0,1,\ww)$&1776\\
${4}$&$(1,\ww,1,0,0,1,\ww,1)$&$(0,\ww,1,0,1,0,1,\ww)$&1344&
${34}$&$(0,1,0,0,0,\vv,\vv,\ww)$&$(\ww,\ww,\vv,\ww,\ww,0,1,\ww)$&1776\\
${5}$&$(1,\ww,1,0,0,\vv,\ww,\vv)$&$(\vv,\vv,\ww,\ww,0,1,1,\ww)$&1344&
${35}$&$(1,1,0,0,0,\vv,\ww,1)$&$(\vv,1,\vv,\ww,1,\ww,1,\ww)$&1776\\
${6}$&$(0,1,0,1,0,\vv,\ww,1)$&$(\ww,\vv,\ww,\ww,1,1,1,\ww)$&1392&
${36}$&$(1,1,1,0,0,\ww,\ww,1)$&$(\vv,1,\ww,\vv,\vv,0,1,\ww)$&1776\\
${7}$&$(1,\vv,0,0,0,1,\ww,1)$&$(1,\ww,0,\vv,1,0,\vv,1)$&1392&
${37}$&$(1,\vv,\vv,0,0,1,\ww,0)$&$(1,\ww,\ww,1,1,\vv,\vv,1)$&1776\\
${8}$&$(0,1,0,1,0,\ww,\ww,\ww)$&$(\vv,\vv,\ww,1,1,1,1,\ww)$&1392&
${38}$&$(1,\vv,\vv,0,0,1,\ww,0)$&$(\ww,\vv,0,1,\ww,0,\vv,1)$&1776\\
${9}$&$(0,1,1,0,0,\ww,\ww,0)$&$(\ww,1,\ww,\ww,1,0,1,\ww)$&1392&
${39}$&$(1,1,1,0,0,1,1,1)$&$(\ww,0,\vv,1,0,0,1,\ww)$&1776\\
${10}$&$(1,\vv,\ww,0,0,1,\ww,\ww)$&$(1,0,\ww,0,\ww,0,\ww,\vv)$&1488&
${40}$&$(1,0,0,0,0,\vv,\vv,1)$&$(\vv,0,\ww,\vv,0,0,1,\ww)$&1776\\
${11}$&$(1,\vv,\vv,0,0,\vv,\ww,\ww)$&$(\ww,\ww,1,\ww,1,0,\vv,1)$&1488&
${41}$&$(1,0,0,0,0,\ww,\ww,1)$&$(\vv,1,\ww,\vv,\vv,0,1,\ww)$&1776\\
${12}$&$(1,0,1,0,0,1,\vv,\ww)$&$(\vv,0,\vv,1,1,0,1,\ww)$&1488&
${42}$&$(1,1,0,0,0,1,1,\ww)$&$(\ww,1,\ww,1,\ww,1,1,\ww)$&1776\\
${13}$&$(0,1,0,0,0,\vv,\vv,\ww)$&$(\ww,\vv,\ww,\ww,\ww,0,1,\ww)$&1488&
${43}$&$(1,\vv,0,0,0,\vv,1,\vv)$&$(\ww,\ww,1,\vv,\vv,\vv,\vv,1)$&1776\\
${14}$&$(0,1,0,0,0,\vv,\ww,0)$&$(\vv,\ww,\vv,1,\vv,1,1,\ww)$&1536&
${44}$&$(1,1,1,0,0,1,1,1)$&$(\vv,1,\vv,\ww,1,0,1,\ww)$&1776\\
${15}$&$(0,1,0,0,0,\ww,1,\vv)$&$(\vv,\ww,\ww,1,1,0,1,\ww)$&1632&
${45}$&$(0,1,0,\vv,0,\ww,\vv,\vv)$&$(\vv,\ww,1,1,1,1,\vv,1)$&1824\\
${16}$&$(1,\vv,\vv,0,0,1,\vv,\ww)$&$(\vv,0,\vv,0,\ww,0,\vv,1)$&1632&
${46}$&$(1,1,1,0,0,1,1,1)$&$(\ww,1,\vv,\vv,\ww,0,1,\ww)$&1824\\
${17}$&$(1,\vv,0,\vv,0,\vv,1,1)$&$(\ww,\vv,1,\vv,1,0,\vv,1)$&1632&
${47}$&$(1,0,\vv,0,0,\ww,\vv,\ww)$&$(\vv,\ww,1,\vv,\ww,\vv,\vv,1)$&1824\\
${18}$&$(0,1,0,\vv,0,\ww,\vv,\vv)$&$(\ww,1,\ww,\vv,\ww,\vv,\vv,1)$&1632&
${48}$&$(1,\vv,0,0,0,1,\ww,1)$&$(\ww,0,\ww,\vv,\vv,0,\vv,1)$&1824\\
${19}$&$(1,0,\vv,0,0,\vv,0,1)$&$(\vv,0,1,0,1,0,\vv,1)$&1632&
${49}$&$(1,1,0,0,0,1,1,\ww)$&$(\ww,0,1,\vv,\vv,0,1,\ww)$&1824\\
${20}$&$(1,1,0,0,0,1,1,\ww)$&$(\vv,1,1,\vv,\ww,1,1,\ww)$&1632&
${50}$&$(0,0,0,0,0,1,1,0)$&$(\ww,1,\vv,\vv,\ww,0,1,\ww)$&1824\\
${21}$&$(0,1,0,0,0,1,\ww,\ww)$&$(\ww,1,\ww,\vv,\ww,0,1,\ww)$&1680&
${51}$&$(1,0,\vv,0,0,\vv,\vv,\vv)$&$(\ww,\ww,1,0,\vv,0,\vv,1)$&1920\\
${22}$&$(0,1,0,0,0,\vv,1,\ww)$&$(\ww,\ww,\vv,\vv,\vv,0,1,\ww)$&1680&
${52}$&$(1,1,0,0,0,\ww,1,0)$&$(\vv,1,0,1,0,0,1,\ww)$&1920\\
${23}$&$(1,1,0,0,0,\vv,1,0)$&$(\vv,\vv,\ww,\vv,\ww,0,1,\ww)$&1680&
${53}$&$(1,0,0,0,0,\ww,1,0)$&$(\ww,0,1,1,1,0,1,\ww)$&1920\\
${24}$&$(0,0,1,0,0,\ww,\ww,\ww)$&$(\vv,1,1,\ww,\ww,0,1,\ww)$&1680&
${54}$&$(1,0,0,0,0,\vv,1,0)$&$(\vv,\vv,1,\ww,1,1,1,\ww)$&1920\\
${25}$&$(1,0,1,0,0,\vv,\ww,1)$&$(1,\ww,1,1,\vv,1,1,\ww)$&1680&
${55}$&$(1,\vv,\vv,0,0,\ww,1,\ww)$&$(1,\vv,\vv,\ww,\vv,0,\vv,1)$&1920\\
${26}$&$(1,1,1,0,0,\vv,1,\ww)$&$(\ww,0,\vv,0,\ww,0,1,\ww)$&1680&
${56}$&$(1,\vv,\vv,0,0,1,\ww,0)$&$(1,1,\vv,1,\vv,\vv,\vv,1)$&1920\\
${27}$&$(1,\vv,\ww,1,0,1,\ww,\vv)$&$(\ww,1,1,1,1,1,1,\ww)$&1680&
${57}$&$(1,\vv,\ww,0,0,\vv,\ww,1)$&$(\vv,1,\ww,\ww,1,0,\ww,\vv)$&1920\\
${28}$&$(0,1,1,1,0,\ww,\vv,\ww)$&$(\ww,1,1,\ww,\vv,0,1,\ww)$&1776&
${58}$&$(0,1,1,0,0,\ww,\ww,0)$&$(\vv,\vv,1,1,\ww,0,1,\ww)$&1920\\
${29}$&$(1,\vv,1,0,0,\vv,1,1)$&$(\ww,0,\ww,0,1,0,1,\ww)$&1776&
${59}$&$(1,0,\vv,0,0,\vv,\vv,\vv)$&$(\vv,\ww,1,\vv,\ww,\vv,\vv,1)$&1968\\
${30}$&$(1,\vv,0,\vv,0,\vv,1,1)$&$(1,\ww,1,\ww,\ww,0,\vv,1)$&1776&
&&&\\
\noalign{\hrule height1pt}
\end{tabular}
}
\end{table}
     %%%%%%%%%%%%%%%%%  table  %%%%%%%%%%%%%%%%%

       %%%%%%%%%%%%%%%%%  table  %%%%%%%%%%%%%%%%%
\begin{table}[thbp]
\caption{Modified four $\ww$-circulant Hermitian self-dual $[32,16,10]$ codes $C_{32,\ww,i}$}
\label{Tab:32-2}
\centering
\medskip
%{\small
{\footnotesize
%{\scriptsize
\begin{tabular}{c|ll|c||c|ll|c}
\noalign{\hrule height1pt}
$i$ & \multicolumn{1}{c}{$r_A$} & \multicolumn{1}{c|}{$r_B$} & $A_{10}$&
$i$ & \multicolumn{1}{c}{$r_A$} & \multicolumn{1}{c|}{$r_B$}  & $A_{10}$ \\
\hline
${1}$&$(0,1,0,0,0,\vv,1,\vv)$&$(\vv,\vv,\vv,\ww,0,1,1,\ww)$&1200&
${31}$&$(1,1,1,0,0,1,\ww,0)$&$(1,\vv,0,\vv,0,1,1,\ww)$&1776\\
${2}$&$(0,1,1,0,0,\vv,\vv,\vv)$&$(0,1,1,\vv,\vv,0,1,\ww)$&1200&
${32}$&$(0,1,0,0,0,1,1,0)$&$(1,1,\vv,\ww,0,0,1,\ww)$&1776\\
${3}$&$(1,\ww,\ww,0,0,1,\ww,0)$&$(\vv,0,\ww,1,0,\vv,\ww,\vv)$&1200&
${33}$&$(0,1,0,0,0,\ww,1,\ww)$&$(\ww,\ww,0,1,1,1,1,\ww)$&1776\\
${4}$&$(1,1,0,0,0,\ww,1,1)$&$(1,\ww,0,\vv,1,0,1,\ww)$&1344&
${34}$&$(0,1,0,0,0,\vv,1,\vv)$&$(\ww,\ww,\vv,1,1,0,1,\ww)$&1776\\
${5}$&$(1,\vv,\vv,0,0,\ww,\ww,1)$&$(1,\vv,1,1,\ww,0,\vv,1)$&1344&
${35}$&$(1,1,0,1,0,\vv,\ww,0)$&$(1,\vv,\ww,\ww,1,1,1,\ww)$&1776\\
${6}$&$(0,1,0,1,0,\ww,\ww,1)$&$(\ww,\ww,\vv,\ww,1,1,1,\ww)$&1392&
${36}$&$(1,\vv,1,0,0,\vv,1,\ww)$&$(\vv,\vv,\vv,\vv,\vv,0,1,\ww)$&1776\\
${7}$&$(1,1,0,0,0,1,\vv,1)$&$(\vv,0,1,1,0,1,1,\ww)$&1392&
${37}$&$(0,1,1,0,0,1,1,\ww)$&$(\ww,1,\vv,1,1,1,1,\ww)$&1776\\
${8}$&$(1,1,0,1,0,\vv,0,1)$&$(\ww,1,1,\ww,1,1,1,\ww)$&1392&
${38}$&$(0,1,1,0,0,1,1,\ww)$&$(\ww,1,0,0,\vv,\ww,1,\ww)$&1776\\
${9}$&$(0,1,1,0,0,1,\vv,0)$&$(\ww,\vv,1,1,\ww,0,1,\ww)$&1392&
${39}$&$(0,1,1,0,0,1,1,\ww)$&$(1,\ww,\vv,1,0,0,1,\ww)$&1776\\
${10}$&$(0,1,0,1,0,\ww,\ww,1)$&$(\vv,1,0,\ww,1,0,1,\ww)$&1488&
${40}$&$(1,0,1,0,0,\vv,0,\vv)$&$(0,1,\ww,\ww,0,0,1,\ww)$&1776\\
${11}$&$(1,\ww,\ww,0,0,\ww,\vv,\vv)$&$(\vv,\ww,0,\vv,\vv,\ww,\ww,\vv)$&1488&
${41}$&$(1,0,0,0,0,1,\ww,\ww)$&$(\ww,1,\vv,\ww,1,0,1,\ww)$&1776\\
${12}$&$(1,0,1,0,0,\ww,1,\ww)$&$(\vv,0,\ww,\ww,\ww,0,1,\ww)$&1488&
${42}$&$(1,\ww,0,0,0,\vv,1,\vv)$&$(\vv,1,1,\vv,\vv,\ww,\ww,\vv)$&1776\\
${13}$&$(0,1,0,0,0,\vv,1,\vv)$&$(1,\vv,\vv,1,1,0,1,\ww)$&1488&
${43}$&$(1,1,0,0,0,\ww,1,1)$&$(1,\ww,1,\vv,1,1,1,\ww)$&1776\\
${14}$&$(0,1,0,0,0,\vv,\ww,0)$&$(\vv,\ww,\ww,\ww,1,1,1,\ww)$&1536&
${44}$&$(1,1,1,1,0,\vv,\ww,1)$&$(1,\ww,\vv,1,0,0,1,\ww)$&1776\\
${15}$&$(0,1,0,0,0,1,\vv,\ww)$&$(\ww,\vv,\ww,\ww,\ww,0,1,\ww)$&1632&
${45}$&$(0,1,0,1,0,1,1,\ww)$&$(1,\ww,1,\vv,\ww,\ww,1,\ww)$&1824\\
${16}$&$(0,1,0,1,0,\ww,\ww,1)$&$(\ww,1,\vv,\ww,0,0,1,\ww)$&1632&
${46}$&$(1,\ww,\vv,0,0,\vv,1,\ww)$&$(1,\vv,1,\vv,0,\vv,\vv,1)$&1824\\
${17}$&$(1,1,0,1,0,\ww,1,1)$&$(1,\vv,1,\ww,1,0,1,\ww)$&1632&
${47}$&$(1,0,1,0,0,\ww,\ww,1)$&$(\vv,\ww,1,\vv,\ww,1,1,\ww)$&1824\\
${18}$&$(0,1,0,1,0,1,1,\ww)$&$(\vv,\ww,\ww,\ww,1,1,1,\ww)$&1632&
${48}$&$(1,1,0,0,0,1,\vv,1)$&$(\vv,0,\ww,\ww,\ww,0,1,\ww)$&1824\\
${19}$&$(1,0,0,0,0,\vv,\vv,1)$&$(1,0,\ww,0,1,0,1,\ww)$&1632&
${49}$&$(1,\ww,0,0,0,\vv,1,\vv)$&$(1,0,\ww,\vv,\ww,0,\ww,\vv)$&1824\\
${20}$&$(1,\ww,0,0,0,\vv,1,\vv)$&$(\vv,\ww,1,\ww,\ww,\ww,\ww,\vv)$&1632&
${50}$&$(0,0,0,0,0,1,\ww,0)$&$(\vv,\ww,\vv,\ww,0,\ww,\ww,\vv)$&1824\\
${21}$&$(0,1,0,0,0,\ww,\ww,\vv)$&$(\ww,\ww,\ww,1,0,1,1,\ww)$&1680&
${51}$&$(1,0,1,0,0,1,\ww,\vv)$&$(\vv,\vv,\ww,0,\vv,0,1,\ww)$&1920\\
${22}$&$(0,1,0,0,0,1,\vv,\ww)$&$(1,\vv,1,\ww,1,0,1,\ww)$&1680&
${52}$&$(1,1,0,0,0,\vv,1,0)$&$(1,\vv,0,0,1,0,1,\ww)$&1920\\
${23}$&$(1,1,0,0,0,\vv,1,0)$&$(\vv,\ww,1,\ww,\vv,0,1,\ww)$&1680&
${53}$&$(1,0,0,0,0,\ww,1,0)$&$(\vv,0,\ww,\vv,1,0,1,\ww)$&1920\\
${24}$&$(0,0,1,0,0,\ww,\vv,1)$&$(\vv,1,1,1,1,0,1,\ww)$&1680&
${54}$&$(1,0,0,0,0,\ww,1,0)$&$(\vv,\ww,1,1,1,1,1,\ww)$&1920\\
${25}$&$(1,0,1,0,0,\vv,1,1)$&$(1,\ww,\vv,\ww,1,1,1,\ww)$&1680&
${55}$&$(1,\vv,\vv,0,0,\ww,\vv,\vv)$&$(\ww,\vv,\ww,0,1,\vv,\vv,1)$&1920\\
${26}$&$(1,1,1,0,0,\ww,\vv,1)$&$(\vv,0,\vv,0,\vv,0,1,\ww)$&1680&
${56}$&$(0,1,1,0,0,1,1,\ww)$&$(\ww,\ww,1,1,\vv,1,1,\ww)$&1920\\
${27}$&$(1,1,1,1,0,\vv,\ww,1)$&$(1,\ww,\vv,1,\vv,1,1,\ww)$&1680&
${57}$&$(1,1,1,0,0,\ww,\ww,\ww)$&$(\ww,1,\ww,1,0,1,1,\ww)$&1920\\
${28}$&$(1,\ww,\ww,\ww,0,\vv,\vv,1)$&$(\ww,0,\vv,\vv,\ww,0,\ww,\vv)$&1776&
${58}$&$(0,1,1,0,0,1,\vv,0)$&$(\ww,1,\ww,1,0,1,1,\ww)$&1920\\
${29}$&$(1,\vv,\vv,0,0,1,1,\ww)$&$(\ww,0,1,0,1,0,\vv,1)$&1776&
${59}$&$(1,0,1,0,0,1,\ww,\vv)$&$(\vv,\ww,1,\vv,\ww,1,1,\ww)$&1968\\
${30}$&$(1,1,0,1,0,\ww,1,1)$&$(\ww,\vv,\ww,\ww,\ww,0,1,\ww)$&1776&
&&&\\
\noalign{\hrule height1pt}
\end{tabular}
}
\end{table}
     %%%%%%%%%%%%%%%%%  table  %%%%%%%%%%%%%%%%%

       %%%%%%%%%%%%%%%%%  table  %%%%%%%%%%%%%%%%%
\begin{table}[thbp]
\caption{Modified four $\vv$-circulant Hermitian self-dual $[32,16,10]$ codes $C_{32,\vv,i}$}
\label{Tab:32-3}
\centering
\medskip
%{\small
{\footnotesize
%{\scriptsize
\begin{tabular}{c|ll|c||c|ll|c}
\noalign{\hrule height1pt}
$i$ & \multicolumn{1}{c}{$r_A$} & \multicolumn{1}{c|}{$r_B$} & $A_{10}$&
$i$ & \multicolumn{1}{c}{$r_A$} & \multicolumn{1}{c|}{$r_B$}  & $A_{10}$ \\
\hline
${1}$&$(0,1,0,0,0,\ww,1,\ww)$&$(1,1,1,\vv,\vv,0,1,\ww)$&1200&
${31}$&$(0,1,1,0,0,\vv,\ww,\vv)$&$(\ww,\ww,\ww,0,\ww,0,1,\ww)$&1776\\
${2}$&$(0,1,1,0,0,\ww,\ww,\ww)$&$(1,0,\vv,\ww,0,1,1,\ww)$&1200&
${32}$&$(0,1,0,0,0,1,1,0)$&$(\ww,\ww,1,\vv,0,0,1,\ww)$&1776\\
${3}$&$(0,1,1,0,0,\ww,\ww,\ww)$&$(\ww,\ww,\vv,1,0,0,1,\ww)$&1200&
${33}$&$(0,1,0,0,0,\vv,\vv,\ww)$&$(1,\vv,1,\ww,\vv,0,1,\ww)$&1776\\
${4}$&$(1,\ww,\ww,0,0,\ww,\vv,\vv)$&$(0,\vv,\vv,0,1,0,\ww,\vv)$&1344&
${34}$&$(0,1,0,0,0,1,\vv,\vv)$&$(\vv,\ww,\vv,\vv,\ww,0,1,\ww)$&1776\\
${5}$&$(1,\ww,\ww,0,0,\vv,\vv,1)$&$(\vv,\ww,\ww,1,0,\vv,\ww,\vv)$&1344&
${35}$&$(1,1,0,0,0,1,\vv,\ww)$&$(1,\ww,\vv,1,1,1,1,\ww)$&1776\\
${6}$&$(0,1,0,1,0,\vv,\vv,1)$&$(\ww,1,1,\ww,1,1,1,\ww)$&1392&
${36}$&$(1,1,\vv,0,0,\ww,1,\vv)$&$(1,\ww,\vv,1,\vv,0,\vv,1)$&1776\\
${7}$&$(1,1,0,0,0,\ww,\vv,\vv)$&$(\vv,1,0,\vv,1,0,1,\ww)$&1392&
${37}$&$(0,1,1,0,0,\ww,1,\ww)$&$(\vv,\vv,\ww,\ww,1,\ww,1,\ww)$&1776\\
${8}$&$(1,1,0,1,0,\ww,0,1)$&$(\vv,\vv,1,\ww,1,1,1,\ww)$&1392&
${38}$&$(0,1,1,0,0,1,1,\vv)$&$(\ww,1,\ww,\vv,0,0,1,\ww)$&1776\\
${9}$&$(0,1,1,0,0,1,\ww,0)$&$(\vv,\ww,\ww,1,0,\ww,1,\ww)$&1392&
${39}$&$(1,\vv,\ww,0,0,\ww,1,\vv)$&$(1,0,1,\ww,0,0,\ww,\vv)$&1776\\
${10}$&$(1,1,1,0,0,\ww,\vv,\ww)$&$(\ww,0,\vv,0,1,0,1,\ww)$&1488&
${40}$&$(1,0,0,0,0,\vv,\ww,\vv)$&$(\ww,0,1,\ww,0,0,1,\ww)$&1776\\
${11}$&$(1,\vv,\vv,0,0,\vv,\ww,\ww)$&$(1,\vv,\vv,1,1,0,\vv,1)$&1488&
${41}$&$(1,0,0,0,0,1,\vv,\vv)$&$(\ww,\vv,1,\ww,1,0,1,\ww)$&1776\\
${12}$&$(1,0,\vv,0,0,1,1,\vv)$&$(1,0,\ww,\vv,\ww,0,\vv,1)$&1488&
${42}$&$(1,\vv,0,0,0,\ww,1,1)$&$(\vv,1,1,\vv,1,\vv,\vv,1)$&1776\\
${13}$&$(0,1,0,0,0,1,\vv,\vv)$&$(\vv,1,1,\vv,\ww,0,1,\ww)$&1488&
${43}$&$(1,1,0,0,0,\vv,1,1)$&$(\ww,1,\vv,\ww,\ww,\ww,1,\ww)$&1776\\
${14}$&$(0,1,0,0,0,\ww,\vv,0)$&$(1,1,\ww,\ww,1,1,1,\ww)$&1536&
${44}$&$(1,\vv,\ww,0,0,\ww,1,\vv)$&$(\vv,1,1,1,1,0,\ww,\vv)$&1776\\
${15}$&$(0,1,0,0,0,\ww,\ww,\ww)$&$(\ww,\ww,1,1,\vv,0,1,\ww)$&1632&
${45}$&$(0,1,0,1,0,\vv,\ww,\ww)$&$(\ww,\vv,\ww,1,\vv,1,1,\ww)$&1824\\
${16}$&$(1,1,1,0,0,\vv,\ww,\ww)$&$(1,0,\ww,0,1,0,1,\ww)$&1632&
${46}$&$(1,\vv,\ww,0,0,\ww,1,\vv)$&$(1,\ww,\vv,\ww,\ww,0,\ww,\vv)$&1824\\
${17}$&$(1,1,0,1,0,\vv,1,1)$&$(\vv,\ww,\vv,\ww,1,0,1,\ww)$&1632&
${47}$&$(1,0,1,0,0,\vv,\vv,1)$&$(\ww,\vv,1,\ww,\ww,1,1,\ww)$&1824\\
${18}$&$(0,1,0,1,0,\vv,\ww,\ww)$&$(1,1,\ww,\ww,1,1,1,\ww)$&1632&
${48}$&$(1,1,0,0,0,\ww,\vv,\vv)$&$(\ww,0,\vv,1,\vv,0,1,\ww)$&1824\\
${19}$&$(1,0,1,0,0,1,0,\vv)$&$(1,0,1,0,\ww,0,1,\ww)$&1632&
${49}$&$(1,\vv,0,0,0,\ww,1,1)$&$(\ww,0,1,1,\vv,0,\vv,1)$&1824\\
${20}$&$(1,\vv,0,0,0,\ww,1,1)$&$(1,\ww,1,1,1,1,\vv,1)$&1632&
${50}$&$(0,0,0,0,0,1,\vv,0)$&$(\ww,\vv,1,\vv,\vv,0,\vv,1)$&1824\\
${21}$&$(0,1,0,0,0,\vv,\vv,\ww)$&$(1,1,\ww,1,1,0,1,\ww)$&1680&
${51}$&$(1,0,1,0,0,1,\vv,\ww)$&$(\vv,\vv,\vv,0,1,0,1,\ww)$&1920\\
${22}$&$(0,1,0,0,0,1,\ww,\vv)$&$(\vv,\ww,\vv,\ww,1,0,1,\ww)$&1680&
${52}$&$(1,1,0,0,0,\ww,1,0)$&$(\ww,\vv,0,1,0,0,1,\ww)$&1920\\
${23}$&$(1,1,0,0,0,\ww,1,0)$&$(1,\ww,1,1,1,0,1,\ww)$&1680&
${53}$&$(1,0,0,0,0,\vv,1,0)$&$(\vv,0,1,1,\ww,0,1,\ww)$&1920\\
${24}$&$(0,0,1,0,0,1,1,\ww)$&$(\ww,1,1,\vv,1,0,1,\ww)$&1680&
${54}$&$(1,0,0,0,0,\vv,1,0)$&$(\ww,\ww,1,1,1,1,1,\ww)$&1920\\
${25}$&$(1,0,1,0,0,\ww,1,1)$&$(\ww,1,\vv,\ww,1,1,1,\ww)$&1680&
${55}$&$(1,\ww,\ww,0,0,\vv,\ww,\ww)$&$(1,1,\vv,\vv,1,0,\ww,\vv)$&1920\\
${26}$&$(1,\vv,\ww,0,0,\vv,1,\ww)$&$(1,0,1,0,\vv,0,\ww,\vv)$&1680&
${56}$&$(0,1,1,0,0,1,1,\vv)$&$(1,\ww,1,1,\vv,1,1,\ww)$&1920\\
${27}$&$(1,1,1,1,0,\ww,\vv,1)$&$(\vv,\ww,1,\vv,\ww,\ww,1,\ww)$&1680&
${57}$&$(1,1,1,0,0,\vv,\vv,\vv)$&$(\vv,1,\ww,1,1,0,1,\ww)$&1920\\
${28}$&$(1,1,1,1,0,1,\vv,\vv)$&$(1,0,\vv,1,1,0,1,\ww)$&1776&
${58}$&$(0,1,1,0,0,1,\ww,0)$&$(\vv,1,\ww,1,1,0,1,\ww)$&1920\\
${29}$&$(1,\vv,1,0,0,\ww,\vv,\vv)$&$(\vv,0,1,0,\vv,0,1,\ww)$&1776&
${59}$&$(1,0,1,0,0,1,\vv,\ww)$&$(\ww,\vv,1,\ww,\ww,1,1,\ww)$&1968\\
${30}$&$(1,1,0,1,0,\vv,1,1)$&$(\ww,\ww,1,1,\vv,0,1,\ww)$&1776&
&&&\\
\noalign{\hrule height1pt}
\end{tabular}
}
\end{table}
     %%%%%%%%%%%%%%%%%  table  %%%%%%%%%%%%%%%%%

      %%%%%%%%%%%%%%%%%  table  %%%%%%%%%%%%%%%%%
\begin{table}[th]
\caption{Modified four $\mu$-circulant Hermitian self-dual codes with large minimum weights}
\label{Tab:L}
\centering
\medskip
%{\small
{\footnotesize
%{\scriptsize
\begin{tabular}{c|c|ll}
\noalign{\hrule height1pt}
Code & $d$ & \multicolumn{1}{c}{$r_A$} & \multicolumn{1}{c}{$r_B$}   \\
\hline
%$C_{40,1}$   &12&$(1,0,0,1,\vv,1,0,0,1,0)$&$(\ww,\ww,1,1,\ww,\ww,0,\ww,\ww,0)$ \\
%$C_{40,\ww}$ &12&$(1,\vv,\vv,1,1,1,\ww,\vv,\vv,\vv)$&$(\ww,1,0,\ww,\vv,\ww,\ww,\vv,0,0)$ \\
%$C_{40,\vv}$ &12&$(1,\ww,\ww,\ww,\vv,\vv,\ww,\vv,0,\ww)$&$(1,\ww,0,1,\vv,0,\vv,0,\ww,0)$ \\
%$C_{44,1}$   &12&$(1,\ww,\vv,0,0,\ww,0,\vv,\ww,1,\ww)$&$(\vv,\ww,\vv,\ww,\ww,\vv,0,\ww,0,0,0)$ \\
%$C_{44,\ww}$ &12&$(1,0,0,\vv,\ww,1,\vv,0,\ww,1,0)$&$(\vv,0,\vv,0,\ww,0,0,\vv,1,\ww,0)$ \\
%$C_{44,\vv}$ &12&$(1,0,\vv,0,1,1,1,\ww,0,\vv,\ww)$&$(\ww,1,\vv,\vv,0,\vv,\vv,\ww,\ww,\vv,0)$ \\
$C_{48,1}$   &14&$(1,1,0,1,\ww,\ww,0,\vv,\vv,1,1,0)$&$(\vv,1,\vv,1,\vv,1,\vv,\vv,\vv,\vv,1,\vv)$ \\
$C_{48,\ww}$ &14&$(0,1,\vv,1,\vv,\ww,\vv,1,\vv,\vv,\vv,0)$&$(\ww,0,\vv,0,1,\ww,\vv,0,1,\vv,\vv,\ww)$ \\
$C_{48,\vv}$ &14&$(0,0,1,1,\vv,\ww,\ww,0,0,\vv,\vv,1)$&$(0,\vv,1,\vv,\vv,\vv,\vv,0,\ww,0,1,\vv)$ \\
$C_{52,1}$   &14&$(1,0,\vv,\vv,1,\vv,1,1,\ww,0,\ww,\ww,1)$&$(\ww,\vv,0,1,\vv,1,\ww,0,1,\ww,1,0,1)$ \\
$C_{52,\ww}$ &14&$(1,0,1,\vv,0,1,1,0,0,\ww,1,0,1)$&$(\vv,\vv,0,1,0,\vv,1,\ww,1,1,\vv,1,\vv)$ \\
$C_{52,\vv}$ &14&$(1,0,\vv,\vv,1,\vv,1,1,\ww,0,\ww,\ww,1)$&$(\vv,0,0,\ww,0,1,1,0,1,\ww,1,0,1)$ \\
$C_{56,1}$   &16&$(0,1,0,0,\ww,\ww,\ww,\vv,1,0,\vv,\ww,\ww,\vv)$&$(1,1,0,0,\vv,\vv,\ww,\vv,\ww,0,\ww,1,1,1)$ \\
$C_{56,\ww}$ &16&$(1,\ww,\vv,0,1,\vv,\ww,1,\ww,0,\vv,\vv,\ww,\vv)$&$(0,0,\vv,\ww,1,0,\vv,0,0,\vv,\ww,\vv,\ww,\vv)$ \\
$C_{56,\vv}$ &14&$(1,\ww,\ww,0,0,\ww,1,0,1,0,1,0,0,\ww)$&$(\ww,0,1,\ww,\ww,0,0,0,0,\ww,\vv,\ww,1,\ww)$ \\
$C_{60,1}$   &16&$(1,0,\ww,1,\ww,\vv,0,1,1,1,\vv,0,1,1,\vv)$&$(\vv,0,\ww,1,0,\ww,\vv,\ww,\vv,0,0,\vv,\vv,\vv,\vv)$ \\
$C_{60,\ww}$ &16&$(1,\vv,1,0,0,0,\vv,\ww,0,\ww,\vv,1,\vv,\ww,\ww)$&$(\vv,0,\ww,0,1,\ww,\ww,0,\vv,1,1,\ww,\vv,1,1)$ \\
$C_{60,\vv}$ &16&$(1,\vv,1,0,0,0,\vv,\ww,0,\ww,\vv,1,\vv,\ww,\ww)$&$(\vv,\ww,1,0,0,\ww,0,0,\ww,0,1,\ww,\vv,1,1)$ \\
$C_{64,1}$   &16&$(1,\vv,\vv,\vv,\vv,0,1,1,1,1,\ww,\vv,\ww,\ww,\ww,\ww)$&$(1,\ww,0,\vv,\ww,1,1,1,\vv,\vv,\vv,\ww,0,1,\ww,\ww)$ \\
$C_{64,\ww}$ &16&$(1,\ww,0,1,1,1,\ww,\vv,\ww,\vv,0,1,0,1,\vv,0)$&$(0,1,1,1,1,0,\vv,\vv,\vv,1,1,\vv,0,1,\vv,1)$ \\
$C_{64,\vv}$ &16&$(1,\ww,0,1,1,1,\ww,\vv,\ww,\vv,0,1,0,1,\vv,0)$&$(0,1,1,1,1,0,\vv,\vv,1,0,\vv,1,\ww,\ww,\ww,\vv)$ \\
$C_{68,1}$   &18&
$(1,0,\vv,\ww,\ww,\vv,1,\vv,0,\vv,0,1,1,1,0,\ww,\vv)$&
$(\ww,1,0,1,\vv,1,\vv,\vv,\vv,\ww,\ww,0,\vv,\ww,\vv,0,\vv)$\\
$C_{68,\ww}$ &18&$(1,\vv,\ww,\ww,\vv,0,1,\vv,\vv,\vv,\ww,0,\vv,\ww,\vv,0,0)$&$(\vv,1,0,1,\ww,\vv,\vv,\ww,1,0,\ww,0,1,0,\ww,\ww,0)$ \\
$C_{68,\vv}$ &18&$(0,1,\ww,0,\ww,0,\ww,\ww,1,\vv,\ww,\ww,\ww,\ww,\vv,1,\vv)$&
$(0,1,\ww,\vv,\ww,0,0,1,1,\vv,1,\vv,1,\vv,\ww,0,\ww)$\\
$C_{72,1}$&18&$(1,\vv,0,\ww,0,\ww,\vv,0,0,0,\ww,\vv,0,\vv,0,\ww,\vv,\ww)$&$(\vv,\ww,0,1,0,\ww,1,0,1,0,\ww,1,1,1,0,0,1,1)$\\
$C_{72,\ww}$&18&$(1,0,\ww,\vv,\ww,0,0,\ww,\ww,0,\vv,1,1,1,1,\ww,1,\vv)$&$(\vv,\vv,1,\vv,1,1,\vv,\ww,\vv,0,\ww,0,1,\ww,1,1,1,0)$\\
$C_{72,\vv}$&18&$(0,1,1,\ww,0,1,\ww,\ww,\ww,1,\vv,1,\vv,1,\ww,0,\ww,0)$&$(1,\vv,\vv,\vv,\vv,0,\ww,0,0,0,\ww,0,1,\vv,\vv,0,1,0)$\\
$C_{76,1}$&18&$(0,1,\vv,1,\vv,0,0,\ww,1,1,\vv,\ww,\ww,\vv,1,\ww,0,\vv,\vv)$&$(0,\ww,1,0,0,\vv,1,\vv,1,0,0,1,\ww,0,\vv,0,\vv,0,0)$\\
$C_{76,\ww}$&18&$(0,1,\vv,\ww,0,1,\ww,1,\vv,\vv,0,0,1,\vv,0,\ww,1,\vv,\vv)$&$(\vv,1,1,1,0,\ww,1,\vv,\vv,1,\vv,0,1,1,1,0,\vv,\ww,0)$\\
$C_{76,\vv}$&18&$(0,1,\ww,\vv,\vv,1,0,0,1,1,0,1,0,\ww,1,1,\vv,\ww,1)$&$(1,\vv,1,1,1,0,0,1,\ww,\ww,0,\ww,1,\vv,\vv,0,1,0,0)$\\
$C_{80,1}$&20&
$(1,\vv,0,\ww,\vv,\vv,\vv,1,\ww,0,\ww,1,1,0,\vv,1,1,\ww,\ww,0)$&
$(1,1,0,0,\vv,1,1,\ww,\vv,0,0,0,\vv,1,\ww,\vv,\vv,1,1,\ww)$\\
$C_{80,\ww}$&20&
$(1,0,\vv,0,\ww,0,1,\vv,\ww,1,\ww,\vv,\vv,\vv,\ww,\vv,\vv,\ww,0,0)$&
$(0,0,\ww,0,1,0,\vv,\vv,\ww,\vv,\vv,\ww,0,\vv,0,\vv,\ww,\vv,\vv,1)$\\
$C_{80,\vv}$&20&
$(0,1,\vv,\ww,\vv,1,\vv,\vv,0,1,\ww,0,\ww,0,1,1,\ww,\vv,1,\ww)$&
$(1,1,\vv,\vv,\ww,\vv,0,0,1,1,\vv,\ww,0,\vv,1,0,1,0,\ww,\ww)$\\
%%%%%%%%%%%%%%%%
$C_{84,\ww}$&20&
$(1,\ww,\vv,\vv,0,0,1,\vv,1,1,\ww,\ww,\vv,0,1,0,\ww,1,1,\ww,1)$&
$(\ww,0,0,0,0,\ww,\vv,0,0,1,1,1,\ww,0,0,\vv,0,0,0,0,0)$
\\
$C_{88,\ww}$&20&
$(0,0,1,0,0,\vv,\vv,\ww,\vv,\vv,0,\ww,1,1,\vv,\ww,\vv,1,\vv,\vv,\vv,0)$&
$(0,0,0,\vv,\ww,\vv,\vv,\ww,\vv,0,1,\vv,\ww,\ww,0,\vv,0,0,0,0,0,0)$
\\
$C_{96,\ww}$&22&
$(1,1,1,\vv,1,\vv,1,0,0,0,0,0,1,0,\ww,\vv,1,0,\vv,\ww,0,\vv,\ww,1)$&
$(1,0,\vv,\ww,\vv,0,1,\vv,0,1,0,\ww,\vv,1,\ww,1,\ww,0,1,\vv,1,1,1,\ww)$
\\
\noalign{\hrule height1pt}
\end{tabular}
}
\end{table}
     %%%%%%%%%%%%%%%%%  table  %%%%%%%%%%%%%%%%%

\end{landscape}

\end{document}